\documentclass[12pt]{article}
\usepackage{sectsty}
\usepackage[T1]{fontenc}
\usepackage{kpfonts}
\usepackage[dvipsnames]{xcolor}
\usepackage{pdfpages}
\usepackage{sgame}
\usepackage{accents}
\usepackage{tikz-cd}
\usepackage{float}
\usepackage[round]{natbib}
\usepackage{multirow}
\usepackage{dutchcal}
\usepackage{pdfpages}
\usepackage{bbm}
\usepackage{sgame}
\usepackage{mathtools}
\usepackage{amsthm}
\usepackage{enumitem}
\usepackage[margin=1.25in]{geometry}
\usepackage{caption}
\usepackage{subcaption}
\usepackage{epigraph}
\usepackage{listings}
\usetikzlibrary{calc}
\usetikzlibrary{shapes,arrows}
\usepackage{tcolorbox}

\newcommand{\ubar}[1]{\underaccent{\bar}{#1}}

\newtheorem{theorem}{Theorem}[section]
\newtheorem{proposition}[theorem]{Proposition}

\newtheorem{corollary}[theorem]{Corollary}
\newtheorem{claim}[theorem]{Claim}

\theoremstyle{definition}

\newtheorem*{definition}{Definition}

\newtheorem{remark}[theorem]{Remark}

\sectionfont{\color{BlueViolet}}
\subsectionfont{\color{BlueViolet}}

\usepackage{hyperref}
\hypersetup{
    colorlinks=true,
    linkcolor=OrangeRed,
    filecolor=CadetBlue,      
    urlcolor=CadetBlue,
    citecolor=Mulberry,
}

\usepackage{setspace}

\definecolor{backcolour}{rgb}{0.63, 0.79, 0.95}

\lstdefinestyle{mystyle}{
  backgroundcolor=\color{backcolour},
  basicstyle=\ttfamily\footnotesize,
  breakatwhitespace=false,         
  breaklines=true,                 
  captionpos=b,                    
  keepspaces=true,                 
  numbers=left,                    
  numbersep=5pt,                  
  showspaces=false,                
  showstringspaces=false,
  showtabs=false,                  
  tabsize=2
}


\begin{document}
\author{Mark Whitmeyer \thanks{Arizona State University, \href{mailto:mark.whitmeyer@gmail.com}{mark.whitmeyer@gmail.com}. Dedicated to KS; may you have found peace. I thank Ralph Boleslavsky, Chris Chambers, Joseph Whitmeyer, and Kun Zhang for their comments. I'm grateful to the department of economics at Duke University for their hospitality during the writing of this draft. Draft date: \today.}}

\title{Comparative Statics for the Subjective}
\maketitle

\begin{abstract}
I study robust comparative statics for risk-averse subjective expected utility (SEU) maximizers. Starting with a finite menu of actions totally ordered by sensitivity to risk, I identify the transformations of her menu that lead a decision-maker to take a lower action, regardless of her particular utility function or belief. My main results reveal that a robust decrease in the action selected is guaranteed by an intuitive steepening/tilting (in belief space) of the actions' payoffs and necessitates a slightly weaker such steepening. This basic pattern generalizes to a broad class of non-EU preferences.
\end{abstract}

\newpage

\section{Introduction}

%Decision-making is central to economics, with myriad applications ranging from consumer behavior to policy design. The study of comparative statics-the analysis of how optimal decisions change in response to transformations in the decision environment--is fundamental to the study of decision-making, \textit{is} the study of decision-making. %How does the exogenous affect the endogenous?

Traditional comparative statics exercises typically focus on changes in exogenous parameters, such as prices or endowments, while holding other structural aspects of the decision problem fixed, e.g., the decision-maker's (DM's) utility function or the form uncertainty takes. Here, we explore how transformations of the monetary payoffs associated with a menu of actions robustly influence the choices of a risk-averse subjective-expected-utility (SEU) maximizer.

In our study, a DM chooses an optimal action given her belief and risk preferences (utility function) from a finite menu. We investigate monotone modifications to the state-dependent \textit{monetary} payoffs of actions within the DM's menu. The key question we address is: under what conditions do such transformations lead a risk-averse DM to systematically choose lower actions from her menu, regardless of her specific utility function (in money) or belief?

We specialize to monotone environments in which the state space is a subset of the real numbers, with higher states yielding higher monetary rewards than lower ones for any fixed action; and the action space is also ordered according to the actions' sensitivities to the DM's risk aversion, \textit{in terms of their state-dependent monetary payments}. As shown by \cite*{safety}, this ordering is equivalent to single-crossing differences, given the structure placed upon the setting.

We begin the analysis by studying the case in which the DM's menu is binary. In Proposition \ref{binaryprop}, we show that a transformation robustly reduces the DM's action if and only if the lower action, \(a\), improves with respect to the higher action, \(b\), in a natural sense: the good outcome for \(a\) must improve, the bad outcome for \(b\) must get worse, and the average monetary reward from \(a\) must improve with respect to \(b\)'s at a risk-neutral DM's indifference belief. At a high level, these conditions can be decomposed into two parts: first, \textit{both} actions are made riskier (in belief space, steeper), then (state-wise) dominance improvements and deteriorations are effected for \(a\) and \(b\).

Moving on to the general finite-action environment, we see that the basic insights from the binary setting persist. In particular, Corollary \ref{corr1} spells out a useful consequence of Proposition \ref{binaryprop}: a collection of pairwise steepenings of the actions' payoffs ensures a robust decrease in the action chosen by the DM. Turning our attention to necessity, Theorem \ref{steeptheorem} reveals that if the DM must choose a lower action after the transformation--irrespective of her risk-averse utility or belief--it must be that for every distinct pair of states and every action that is not inferior to some other action in that pair of states--every \textit{relevant} action--all but the lowest action must become worse, money-wise, in the low state, and all but the highest action must improve in the high state.

In short, Theorem \ref{steeptheorem} reveals that if a transformation robustly leads to a lower action, it must be that in every pair of states, \textit{every} relevant action is made riskier, with the possible exception of only the lowest and highest action. There is a small gap between necessity and sufficiency: this arises because the necessity result concerns only relevant actions, a smaller set than those disciplined by the sufficiency result. Intuitively, relevant actions may nevertheless be dominated by mixtures of actions when the DM is not very risk averse, and become undominated only for high levels of risk aversion. Thus, we only need to guarantee a lower action for a restricted set of utility functions, hence the over-strictness of the proposed sufficient conditions.

We go on in \(\S\)\ref{beyondeu} to observe that our basic insight--that robust decreases in action choices are guaranteed by making monetary rewards more sensitive to the state--persists under far more general preferences than expected utility. We look at a DM whose preferences over actions satisfy basic desiderata (e.g., convexity and monotonicity) but may not be probabilistic. In Corollary \ref{corr43}, we show that, nevertheless, a steepening of the actions' monetary rewards effects a robust decrease in the action selected. We finish the paper with several applications.
%This increase in riskiness plus a collection of pairwise relative improvements and deteriorations of the low and high actions is sufficient for a menu transformation to robustly lead to a lower choice by the decision maker. We spell out one such collection explicitly, understanding that in general it is too strong, i.e., not necessary. 

\subsection{Related Papers}

The advent of the 1990s brought with it a flurry of comparative statics papers, including the seminal works by \cite*{milgrom1994comparing}, \cite*{milgrom1994comparing2}, \cite*{edlin1998strict}, and \cite*{milgrom1994monotone}. The last of these arguably provides the modern foundation for this literature, demonstrating that monotone environments provide just enough structure for establishing clean general results.

The tools introduced by these works (and their forebears, e.g., \cite*{topkis1978minimizing}) have since been widely adopted in economic theory to analyze how structural features of optimization problems affect behavior. Of course, randomness is a central feature of many settings, and so a number of important works study monotone comparative statics in stochastic environments. \cite*{athey2002monotone} is one such paper, as is \cite*{quah2009comparative}. Importantly, in these and the others studying comparative statics under uncertainty, the randomness is a \textit{primitive}, and exogenously specified.

Such specified stochasticity is not present in \cite*{quah2012aggregating}. Instead, they (for all intents and purposes) inhabit an SEU setting, with a random state of the world, and ask when the single-crossing differences property of \cite*{milgrom1994monotone} holds, and so monotone comparative statics obtain, for any distribution over states. \cite*{kartik2023single} tackles a similar problem: as they note, \cite*{quah2012aggregating} study ``...when single crossing is preserved by positive linear combinations. On the other hand, our [\cite*{kartik2023single}] problem turns on arbitrary linear combinations preserving single crossing.''

In short, \cite*{quah2012aggregating} and \cite*{kartik2023single} are robust comparative statics exercises, where the robustness is distributional. \cite*{safety} impose this robustness as well and, notably, show that in monotone environments, ``high'' actions are precisely those that become relatively less attractive as a DM becomes more risk averse. Here, in contrast to these works, we ask for robustness both with respect to a DM's belief (distributional) and utility function (preferences).

Three other papers are closely related to this one. In \cite*{hmab}, a coauthor and I fix a binary menu of actions, \(a\) and \(b\), and ask what transformations of \(a\) make it robustly more attractive than \(a\) versus \(b\). In \cite*{calib}, I conduct a related exercise to \cite*{hmab}'s in order to scrutinize \cite*{rabin2000risk}'s calibration paradox. Special mention is due to \cite{chompchomp}, who conducts a binary-menu study analogous to \(\S\)\ref{binaryenviron}, but when beliefs are known. In short, his exercise is one in which specified \textit{lotteries} are altered; in stark contrast to this paper, in which the agent's belief, which generates the randomness, is a free parameter. We discuss his work further in \(\S\)\ref{binarymenu}.

\section{The Environment}

We study \textcolor{Rhodamine}{Monotone Decision Problems},
\(\left(A,\Theta\right)\), consisting of a compact set of states \(\Theta \subset \mathbb{R}\); and a finite set of actions, \(A\).  \(\Delta \equiv \Delta \left(\Theta\right)\) is the set of all Borel probability measures on \(\Theta\). Each action \(a \in A\) is a bounded, increasing function \(a \colon \Theta \to \mathbb{R}\), with \(a_\theta\) denoting \(a(\theta)\). We stipulate that no action state-wise dominates another: for any distinct \(a, b \in A\), there exist distinct states \(\theta, \theta' \in \Theta\) such that \(a_{\theta} > b_{\theta}\) and \(a_{\theta'} < b_{\theta'}\).\footnote{Henceforth, we omit the ``distinct'' modifier when it is obvious.} Finally, we assume that the actions in \(A\) are totally ordered by single-crossing: \(b \triangleright a\) if 
\[\quad b_{\theta} \underset{(>)}{\geq} a_{\theta} \ \Rightarrow \ b_{\theta'} \underset{(>)}{\geq} a_{\theta'} \text{ for any } \theta' > \theta\text{.}\]

Action set \(A\) undergoes a monotone, order-preserving \textcolor{Rhodamine}{Transformation} if for each action \(a \in A\) its payoff in every state \(\theta\) is transformed from \(a_\theta\) to \(\hat{a}_\theta\), where for all \(a, b \in A\), for all \(\theta, \theta' \in \Theta\),
\[a_{\theta'} \underset{(>)}{\geq} a_{\theta} \ \Rightarrow \ \hat{a}_{\theta'} \underset{(>)}{\geq} \hat{a}_{\theta}, \qquad \text{and} \qquad b_{\theta} \underset{(>)}{\geq} a_{\theta} \ \Rightarrow \ \hat{b}_{\theta} \underset{(>)}{\geq} \hat{a}_{\theta}\text{.}\]
That is, the transformations under consideration preserve the ordinal rankings of monetary payoffs both between actions, for every fixed state; and between states, for every fixed action.

There is a decision-maker (DM), who is a risk-averse subjective-utility maximizer, with a belief \(\mu \in \Delta\) and a continuous, strictly increasing, and concave utility function, \(u \colon \mathbb{R} \to \mathbb{R}\). \(\mathcal{U}\) denotes the class of such utility functions. Given the DM's belief \(\mu \in \Delta\) and utility function \(u \in \mathcal{U}\), before her set of actions \(A\) is transformed, her set of optimal actions is 
\[A^* = A^*\left(\mu,u\right) \coloneqq \left\{a \in A \colon \mathbb{E}_\mu\left[u\left(a_{\theta}\right)\right] \geq \max_{a' \in A} \mathbb{E}_\mu\left[u\left(a_{\theta}'\right)\right]\right\}\text{.}\]
After \(A\) is transformed, her set of optimal actions is
\[\hat{A}^* = \hat{A}^*\left(\mu,u\right) \coloneqq \left\{a \in A \colon \mathbb{E}_\mu\left[u\left(\hat{a}_{\theta}\right)\right] \geq \max_{a' \in A} \mathbb{E}_\mu\left[u\left(\hat{a}_{\theta}'\right)\right]\right\}\text{.}\]
\begin{definition}
    \(A^*\) \textcolor{Rhodamine}{Dominates} \(\hat{A}^*\) \textcolor{Rhodamine}{in the Strong Set Order} if for any \(a' \in \hat{A}^*\) and \(a \in A^*\), we have \(\max\left\{a,a'\right\} \in A^*\) and \(\min\left\{a,a'\right\} \in \hat{A}^*\).
\end{definition}
\begin{definition}
    A transformation \textcolor{Rhodamine}{Reduces the DM's Action} if for any \(\mu \in \Delta\) and \(u \in \mathcal{U}\), \(A^*\) dominates \(\hat{A}^*\) in the strong set order.
\end{definition}
In order to avoid needless repetition, henceforth, when we say a pair of states \(\theta, \theta' \in \Theta\), we understand them not only to be distinct but with \(\theta < \theta'\). 

Before going on to the analysis, it is useful to explicitly spell out the exercise at hand. A typical comparative statics paper, like \cite{topkis1978minimizing}, \cite{milgrom1994monotone}, or \cite{quah2009comparative}, produces theorems showing conditions on primitives that makes optima move up or in some predictable way. This is our exercise as well: the primitives we are altering are the state-dependent monetary payoffs to the actions. This formulation allows us to ask to an extreme form of robustness in ``predictable movement:'' the decrease in optimizers must hold for all beliefs and risk-averse utilities.

%We also maintain the convention that hatted objects are those considered after the transformation; and objects without hats, 

%\subsection{Discussion of the Setup}

\section{The Analysis}

\subsection{A Binary Menu}\label{binarymenu}

We begin with a two-action menu, \(\left\{a,b\right\}\), with \(b \triangleright a\). Recall that this means that \(b\) single-crosses \(a\) (from below). We introduce the notation \[\mathcal{A} \coloneqq \left\{\theta \in \Theta \colon a_{\theta} > b_{\theta}\right\} \quad \text{and} \quad \mathcal{B} \coloneqq \left\{\theta' \in \Theta \colon a_{\theta'} < b_{\theta'}\right\} \text{.}\]
\(\mathcal{A}\) and \(\mathcal{B}\) are the sets of states in which \(a\) yields a strictly higher payoff than \(b\) and vice versa. Naturally, single-crossing implies that \(\mathcal{A}\) consists of low states and \(\mathcal{B}\) consists of high states: \(\sup\left\{\theta \colon \theta \in \mathcal{A}\right\} \leq \inf\left\{\theta' \colon \theta' \in \mathcal{B}\right\}\).
\begin{definition}\label{steepdef}
    Fix actions \(a\) and \(b\) with \(b \triangleright a\). We say that they are \textcolor{Rhodamine}{Made Steeper} if for all \(\theta \in \mathcal{A}\) and \(\theta' \in \mathcal{B}\),
    \begin{enumerate}
        \item \(\hat{b}_{\theta} \leq b_{\theta}\);
        \item \(a_{\theta'} \leq \hat{a}_{\theta'}\); and
        \item A \textbf{Super-Actuarial Improvement} of \(a\) versus \(b\) occurs:
        \[\label{in1}\tag{\(1\)}\frac{\min\left\{\hat{a}_\theta, a_\theta\right\} - \hat{b}_\theta}{a_\theta - b_\theta} \geq \frac{\hat{b}_{\theta'} - \hat{a}_{\theta'}}{\min\left\{\hat{b}_{\theta'}, b_{\theta'}\right\}-a_{\theta'}}\text{.}\]
    \end{enumerate} \end{definition}

\begin{proposition}\label{binaryprop}
    With a two-action menu, \(\left\{a,b\right\}\), a transformation reduces the DM's action if and only if \(a\) and \(b\) are made steeper.
\end{proposition}
The proof lies in Appendix \ref{binarypropproof}. For any pair \(\left(\theta,\theta'\right) \in \mathcal{A} \times \mathcal{B}\), we see that action \(b\) must become worse in the low state (\(\theta\)) and action \(a\) must improve in the high state. In fact, the proposition reduces to first tilting \(a\) and \(b\) ``up''--by raising the payoffs to each in the high state and lowering the payoffs to each in the low state--then effecting a (state-wise) dominance improvement of \(a\) and worsening of \(b\).

Superficially, the sufficiency direction of this proposition appears to follow from the fact that under expected utility, preferences are convex. That is, first take a convex combination of actions \(a\) and \(b\), rotate \(a\) ``toward'' \(b\)--for some \(\lambda \in \left[0,1\right]\), in every \(\theta\), take payoffs \(\lambda a_\theta + (1-\lambda) b_\theta\)--then rotate \(b\) ``away'' from the new \(a\), before finally effecting state-wise dominance improvements and deteriorations of \(a\) and \(b\), respectively. And indeed \textit{this is sufficient}; in fact, as we see in \(\S\)\ref{beyondeu}, it is sufficient for preferences far more general than EU. Nor does this argument make use of our monotone environment.

However, the conditions in the proposition are weaker than those appealed to in the previous paragraph: in the proposition, we make use of the extra structure afforded by expected utility and monotonicity. Namely, expected utility allows us to first jointly translate \(a\) and \(b\) vertically, before doing the rotations toward and away mentioned in the previous paragraph. We use monotonicity in our proof of necessity as it greatly reduces the number of cases we need to check in our proof by contraposition. 

It is also worth highlighting that this proposition is true even \textit{if we do not assume that the DM's utility function is state-independent.} Likewise, our belief-robust conception makes our findings apply in an obvious sense to incomplete preferences \`{a} la \cite{bewley2002knightian}.

\subsubsection{Known Lotteries}

Our study demands robustness along two dimensions, beliefs and utilities. One could instead fix the DM's belief and ask for robust comparative statics over all utility functions within a certain class. \cite{chompchomp} does exactly that in the binary-menu environment. Specifically, for four lotteries (cdfs) \(L_1\), \(L_2\), \(\hat{L}_1\) and \(\hat{L}_2\), he asks when it is it the case for a risk-averse DM that whenever \(L_1 \succeq L_2\), \(\hat{L}_1 \succeq \hat{L}_2\).\footnote{\cite{chompchomp} also characterizes this implication for DM's whose utilities are strictly increasing but not necessarily concave.} In the binary-menu environment, our research question (what transformations reduce the DM's action), in turn, can be rephrased as: for four single-crossing and monotone actions \(a\), \(b\), \(\hat{a}\) and \(\hat{b}\), when it is it the case for a risk-averse DM that whenever \(a \succeq b\), \(\hat{a} \succ \hat{b}\)?

In his Theorem 2, \cite{chompchomp} shows that his desired implication is equivalent to (taking the support of the lotteries to be \(\left[\ubar{x},\bar{x}\right]\)), for all \(x \in \left[\ubar{x},\bar{x}\right]\),
\[\tag{\(2\)}\label{in2}\int_{\ubar{x}}^x\left(L_1(s) - L_2(s)\right)ds \leq 0 \quad \Rightarrow \quad \int_{\ubar{x}}^x\left(\hat{L}_1(s) - \hat{L}_2(s)\right)ds \leq 0\text{,}\]
and
\[\tag{\(3\)}\label{in3}
\sup_{x \in X_{-}} \frac{\int^{x}_{\ubar{x}} \left(\hat{L}_1(s) - \hat{L}_2(s) \right) ds}{\int^{x}_{\ubar{x}}\left( L_1(s) - L_2(s) \right) ds}
\leq 
\inf_{x \in X_{+}} \frac{\int^{x}_{\ubar{x}} \left(\hat{L}_1(s) - \hat{L}_2(s) \right) ds}{\int^{x}_{\ubar{x}}\left( L_1(s) - L_2(s) \right) ds}\text{,}
\]
where
\[X_{-} \coloneqq \left\{ x \in X \,\middle|\, \int^{x}_{\ubar{x}}\left(\hat{L}_1(s) - \hat{L}_2(s) \right) ds > 0 \right\}, 
\text{ and }
X_{+} \coloneqq \left\{ x \in X \,\middle|\, \int_{\ubar{x}}^{x} \left(L_1(s) - L_2(s) \right) ds < 0 \right\}\text{.}
\]
In this paper's SEU environment, given any fixed belief \(\mu\), the actions \(a\) and \(b\) induce lotteries \(L^\mu_a\) and \(L^\mu_b\), pre-transformation, and \(\hat{L}^\mu_a\) and \(\hat{L}^\mu_b\), post-transformation. Thus, actions \(a\) and \(b\) being made steeper (Definition \ref{steepdef}) is equivalent to the induced lotteries satisfying the implication in Expression \ref{in2} and Inequality \ref{in3}, \textit{for any belief} \(\mu \in \Delta\).

%In our SEU environment, given any fixed belief \(\mu\), the lotteries induced by the belief and the actions \(a\) and \(b\), \(L^\mu_a\) and \(L^\mu_b\), are single-crossing; likewise, the lotteries \(\hat{L}^\mu_a\) and \(\hat{L}^\mu_b\). Thus, actions \(a\) and \(b\) being made steeper (Definition \ref{steepdef})

\subsection{A Partial Order on Behavioral Predictions}

The work of \cite{chompchomp} also suggests the following application of Proposition \ref{binaryprop}. In monotone environments, we can define a binary relation on binary menus of ordinally-equivalent actions \(M_1 \coloneqq\left\{a^1,b^1\right\}\) and \(M_2 \coloneqq \left\{a^2,b^2\right\}\), with \(b^1 \triangleright a^1\) (and so by ordinal equivalence, \(b^2 \triangleright a^2\)), \(\underline{\blacktriangleright}\), where \(M_1 \underline{\blacktriangleright} M_2\) if the DM preferring \(a^1\) to \(b^1\) places greater restrictions on rational behavior by the DM than preferring \(a^2\) to \(b^2\). That is, \(M_1 \underline{\blacktriangleright} M_2\) means there is no additional predictive power to knowing \(a^2 \succeq b^2\) when an analyst knows \(a^1 \succeq b^1\) (as \(a^1 \succeq b^1\) \(\Rightarrow\) \(a^2 \succeq b^2\)).

\(M_1 \underline{\blacktriangleright} M_2\) is exactly reducing the DM's action; thus,
\begin{corollary}\label{corrchomp}
    \(M_1 \underline{\blacktriangleright} M_2\) if and only if \(M_2\) is obtained from \(M_1\) by making the actions in \(M_1\) steeper.
\end{corollary}

This implication of the proposition is worth discussing. In a natural sense, menu \(M_2\) is a riskier menu than \(M_1\): both \(a_2\) and \(b_2\) vary more in the state than \(a_1\) and \(b_1\). Moreover, provided \(a_2\) and \(b_2\) are not shifted vertically (in belief space) too much from \(a_1\) and \(b_1\), they are ``safer'' than \(a_1\) and \(b_1\), in the parlance of \cite{safety}; \textit{viz.,} \(a_1\) and \(b_1\) become more appealing than \(a_2\) and \(b_2\) as a DM is made more risk averse. Thus, the corollary tells us that preferring the less risky option in the less risky menu ``has more bite'' than preferring the less risky option in the riskier menu.

As we are in the SEU world, the preference of one action over another in a menu has a joint implication on the DM's possible beliefs and utilities. Preferring \(a_1\) to \(b_1\) is, therefore, a stronger statement than preferring \(a_2\) to \(b_2\) about the DM's pessimism--the \(a_i\)s yield higher monetary payoffs than the \(b_i\)s in the low states, and vice versa--and risk aversion.

\subsection{Beyond Binary: A Sufficient Condition}

\iffalse
For pairs of states \(\theta, \theta' \in \Theta\) and actions \(a, b \in A\) with \(a_\theta > b_\theta\) and \(b_{\theta'} > a_{\theta'}\), we denote \[\mu_{a,b}^{\theta,\theta'} \coloneqq \frac{a_\theta - b_\theta}{a_\theta - b_\theta + b_{\theta'} - a_{\theta'}}\text{.}\] This is the indifference point along the edge of the probability simplex between \(\theta\) and \(\theta'\) at which the DM is indifferent between \(a\) and \(b\).

For an arbitrary action \(a \in A\), let \(A^{a}_{\triangleright}\) denote the set of actions in \(A\) that are larger than it:
\[A^{a}_{\triangleright} \coloneqq \left\{b \in A \colon b \triangleright a\right\}\text{.}\]
For \(\theta, \theta' \in \Theta\), we define the \textcolor{Rhodamine}{\(\left(\theta, \theta'\right)\)-Pertinent Set} \(P^a_{\theta,\theta'}\) as
\[P^a_{\theta,\theta'} \coloneqq \left\{b \in A^{a}_{\triangleright} \colon b_\theta < a_\theta \wedge b_{\theta'} > a_{\theta'}\right\}\text{.}\]
The actions in the \(\left(\theta, \theta'\right)\)-Pertinent Set are those that are relevant for comparative statics when beliefs are restricted to be along the edge between \(\theta\) and \(\theta'\). They are the actions for which there exists some belief on the relative interior of the edge of the simplex between \(\theta\) and \(\theta'\) at which the DM is indifferent between that action and \(a\).
\fi

We now turn our attention to general finite menus. For an arbitrary action \(a \in A\), let \(A^{a}_{\triangleright}\) denote the set of actions in \(A\) that are larger than it:
\[A^{a}_{\triangleright} \coloneqq \left\{b \in A \colon b \triangleright a\right\}\text{.}\] We then deduce a corollary of Proposition \ref{binaryprop}:
\begin{corollary}\label{corr1}
    A transformation reduces the DM's action if for all actions \(a \in A\) and for all actions \(b \in A^{a}_{\triangleright}\), \(a\) and \(b\) are made steeper.
\end{corollary}
To guarantee that a robustly lower action is taken, we need merely ask that a certain collection of pairwise improvements of low versus high actions manifest. It suffices to make every action steeper, with possible dominance improvements and deteriorations at especially low or high actions.

\subsection{Beyond Binary: A Necessary Condition}

For a pair of states \(\theta, \theta' \in \Theta\), we define the set of \textcolor{Rhodamine}{Relevant} actions \(A_{\theta,\theta'} \subseteq A\) as
\[A_{\theta,\theta'} \coloneqq \left\{a \in A \colon \nexists b \in A \colon b_{\theta} \geq a_{\theta} \quad \text{and} \quad b_{\theta'} \geq a_{\theta'}, \text{ with at least one inequality strict}\right\}\text{.}\]
\textit{Viz.,} these are the actions that are not weakly dominated by any other action in \(A\), except possibly by an exact duplicate, in the states \(\theta\) and \(\theta'\). Furthermore, without loss of generality, we assume that for any pair \(\theta, \theta' \in \Theta\) there are no distinct \(a, b \in A_{\theta,\theta'}\) with \(a_\theta = b_\theta\) and \(a_{\theta'} = b_{\theta'}\) (no duplicate actions). % With some abuse of jargon when this holds, we say that there are 

%there are no duplicate actions, i.e., no two distinct actions \(a\) and \(b\) with \(a_\theta\).

We label the actions in \(A_{\theta,\theta'}\) according to \(\triangleright\): 
\[A_{\theta,\theta'} = \left\{a^1, \dots, a^m\right\}\text{, where } a^1 \triangleleft \cdots \triangleleft a^m \text{.}\]
Observe that this implies
\[a^1_{\theta} > \dots > a^m_{\theta} \quad \text{and} \quad a^1_{\theta'} < \dots < a^m_{\theta'}\text{.}\]

For any \(\theta, \theta' \in \Theta\) and any \(A_{\theta,\theta'}\), we let \(\underline{a}^{\left(\theta,\theta'\right)}\) and \(\overline{a}^{\left(\theta,\theta'\right)}\) denote the minimal and maximal actions in \(A_{\theta,\theta'}\): 
\[\underline{a}^{\left(\theta,\theta'\right)} = a^1 \triangleleft \cdots \triangleleft a^m = \overline{a}^{\left(\theta,\theta'\right)} \text{.}\]
Naturally, they may not be distinct; \textit{viz.,} \(A_{\theta,\theta'}\) may be a singleton.

\begin{definition}
    We say that the decision problem \textcolor{Rhodamine}{Becomes Relevantly Steeper} if for any pair \(\theta, \theta' \in \Theta\), \(\hat{a}^i_\theta \leq a^i_\theta\) for all \(a^i \in A_{\theta,\theta'} \setminus \left\{\underline{a}^{\left(\theta,\theta'\right)}\right\}\) and \(\hat{a}^i_{\theta'} \geq a^i_{\theta'}\) for all \(a^i \in A_{\theta,\theta'} \setminus \left\{\overline{a}^{\left(\theta,\theta'\right)}\right\}\).
\end{definition}

\begin{theorem}\label{steeptheorem}
    A transformation reduces the DM's action only if the decision problem becomes relevantly steeper.
\end{theorem}
Please see Appendix \ref{steepproof} for the proof of this theorem. What this result means is that for every pair of states all but the lowest relevant action must become (weakly) worse in the low state, and all but the highest relevant action must become (weakly) better in the high state. The contrapositive of the result provides the intuition: if some non-minimal action (\(a^{i+1}\)) strictly improves in the low state, we can always find a utility function that makes this improvement improve the attractiveness of this action versus all relevant actions lower than it to a large degree. The strict increase means there is a gap between the old value and the new value, so we construct a utility function that is kinked in that gap. As this downward drop in slope becomes increasingly severe, the action immediately below \(a^{i+1}\), \(a^{i}\), becomes increasingly appealing versus \textit{every} action higher than it (including \(a^{i+1}\)); so much so that the DM can be arbitrarily optimistic that the state is \(\theta'\) yet still prefer \(a^i\) to all those higher.

Moreover, the indifference beliefs between \(a^i\) and each \textit{lower} action, pre- and post-transformation, are precisely the risk-neutral indifference beliefs as these monetary rewards all lie above the kink--on which region the constructed utility function is linear--and these beliefs are all strictly bounded away from \(1\). Crucially, the indifference belief between \(i\) and \(i+1\), post-transformation, is also strictly bounded away from \(1\), as it is just the risk-neutral indifference belief. Thus, we can find a sufficiently optimistic belief such that \(a^i\) is optimal pre-transformation, yet strictly inferior to \(a^{i+1}\) post-transformation.

The same essential argument works for establishing the necessity that all but the highest relevant action must become (weakly) better in the high state for a robust decrease in the DM's action. If not, there is some gap we can exploit with a carefully chosen utility function. Finally, it is important to note that the structure placed on the environment, as well as the order-preserving aspect of the transformation, are both used in the proof. It is these features that allow us to construct the single-kinked utility function that is so helpful.

\section{Beyond Expected Utility}\label{beyondeu}

Slightly stronger conditions ensure a robustly lower choice of action for the following abstract setting. Importantly, as we will see, the monotone environment of this paper plays no role here, except to orient the transformed actions as ``steeper'' than the originals. We use a similar formulation of preferences to \cite*{rigotti2008subjective}.\footnote{\cite*{renou2014ordients} conduct comparative statics in a similarly abstract setting, focusing on a generalized first-order approach.} We now understand each action \(a \in A\) to be an element of \(\mathbb{R}^{\Theta}\), and let \(a \gg b\) mean that \(a_\theta > b_\theta\) for all \(\theta \in \Theta\). We impose that the DM has a preference relation over actions \(\succeq\) that is
\begin{enumerate}
    \item \textbf{Strongly Monotone:} If \(a \gg b\) then \(a \succ b\);
    \item \textbf{Convex:} If \(a \succeq b\) then for all \(\lambda \in [0,1]\), \(\lambda a + (1-\lambda) b \succeq b\);
    \item \textbf{Complete:} For all \(a, b \in \mathbb{R}^{\Theta}\), \(a \succeq b\) or \(b \succeq a\); and 
    \item \textbf{Transitive:} If \(a \succeq b\) and \(b \succeq c\) then \(a \succeq c\).
\end{enumerate}
We call such preferences over actions \textcolor{Rhodamine}{Regular}. If \(a \gg b\), we say that \(a\) \textcolor{Rhodamine}{Dominates} \(b\). We term an action \(a^{\lambda} = \lambda a + (1-\lambda) c \succeq b\) (for some \(\lambda \in [0,1]\)) a \textcolor{Rhodamine}{Mixture} of actions \(a\) and \(c\). Observe that strong monotonicity yields that \(a\) dominating \(b\) implies that \(a \succ b\).
\begin{definition}
    Fix actions \(a\) and \(b\) with \(b \triangleright a\). We say that they are \textcolor{Rhodamine}{Made Commonly Steeper} if
    \begin{enumerate}
        \item \(b\) dominates a mixture of \(a\) and \(\hat{b}\) or \(b = \hat{b}\); and
        \item \(\hat{a}\) dominates a mixture of \(a\) and \(\hat{b}\) or \(\hat{a} = a\).
    \end{enumerate} \end{definition}

Deferring its proof to Appendix \ref{regularpropproof}, we have
\begin{proposition}\label{regularprop}
Suppose a DM's preferences over actions are regular. Then, a transformation reduces the DM's action if \(a\) and \(b\) are made commonly steeper.
\end{proposition}

Recall that for an arbitrary action \(a \in A\), \(A^{a}_{\triangleright} = \left\{b \in A \colon b \triangleright a\right\}\). Then,
\begin{corollary}\label{corr43}
    Suppose a DM's preferences over actions are regular. A transformation reduces the DM's action if for any action \(a \in A\) and \(b \in A^{a}_{\triangleright}\), \(a\) and \(b\) are made commonly steeper.
\end{corollary}
Naturally, if the DM is an expected utility (EU) maximizer, with some prior \(\mu \in \Delta\) and a strictly increasing, concave utility function \(u\), her preferences over actions are regular (as defined above).
Here are some other examples of non-EU preferences that are regular (verified in Appendix \ref{regsec}):

\bigskip

\noindent \textbf{Variational preferences.} Following \cite*{maccheroni2006ambiguity}, suppose a DM solves \[\max_{a \in A} \min_{p \in \Delta}\left\{\mathbb{E}_p[u(a_\theta)] + c(p)\right\}\text{,}\]
where \(c \colon \Delta \to \mathbb{R}_{+}\) is a convex function and \(u\) is concave and strictly increasing.

This class of preferences includes as a special case the \textbf{Multiplier preferences} of \cite*{hansen2001robust}. With those, given some \(q \in \Delta\), a DM solves
\[\max_{a \in A} \min_{p \in \Delta}\left\{\mathbb{E}_p[u(a_\theta)] + C_\varphi(p \parallel q)\right\}\text{,}\]
where \(C_\varphi\) is a prior-dependent distortion, defined as
\[
C_\varphi(p \parallel q) = \int_\Theta \varphi\left(\frac{p(\theta)}{q(\theta)}\right) q(\theta) \, d\theta\text{,}
\]
for \(\varphi\)-divergence; and \(u\) is strictly increasing and weakly concave.

\bigskip

\noindent \textbf{Smooth ambiguity.} Following \cite*{klibanoff2005smooth}, a DM evaluates an action \(a\) according to 
\[\int_{\Delta} \varphi \left(\int_{\Theta} u(a(\theta))d\mu(\theta)\right)d\pi\text{,}\]
where \(u\) and \(\varphi\) are strictly increasing and \(u\) and \(\varphi\) are concave.

\bigskip

If there exists a continuous, increasing, and concave \(U\) that represents \(\succeq\), the sufficient conditions for a reduction in the DM's action can be weakened slightly. We call such preferences over actions \textcolor{Rhodamine}{Strongly Convex}.
\begin{definition}\label{mwcs}
    Fix actions \(a\) and \(b\) with \(b \triangleright a\). We say that they are \textcolor{Rhodamine}{Made Weakly Commonly Steeper} if \(b \geq \lambda a + (1-\lambda) \hat{b}\) for some \(\lambda \in [0,1)\) and \(\hat{a} \geq \gamma a + (1-\gamma) \hat{b}\) for some \(\gamma \in [0,1]\). \end{definition}
\begin{proposition}
    Suppose a DM's preferences over actions are strongly convex. A transformation reduces the DM's action if for any action \(a \in A\) and \(b \in A^{a}_{\triangleright}\), \(a\) and \(b\) are made weakly commonly steeper.
\end{proposition}
\begin{proof}
    If \(a \succ \ (\succeq) \ b\), \[U(a) \underset{\geq}{>} U(b) \geq U(\lambda a + (1-\lambda) \hat{b}) \geq \lambda U(a) + (1-\lambda) U\left(\hat{b}\right)\text{,}\]
so \(a \succ \ (\succeq) \ \hat{b}\). Then,
\[U(\hat{a}) \geq U(\gamma a + (1-\gamma)\hat{b}) \geq \gamma U(a) + (1-\gamma) U(\hat{b}) \underset{\geq}{>} U(\hat{b})\text{,}\]
as desired. \end{proof}

\begin{figure}
    \centering
    \includegraphics[width=1\linewidth]{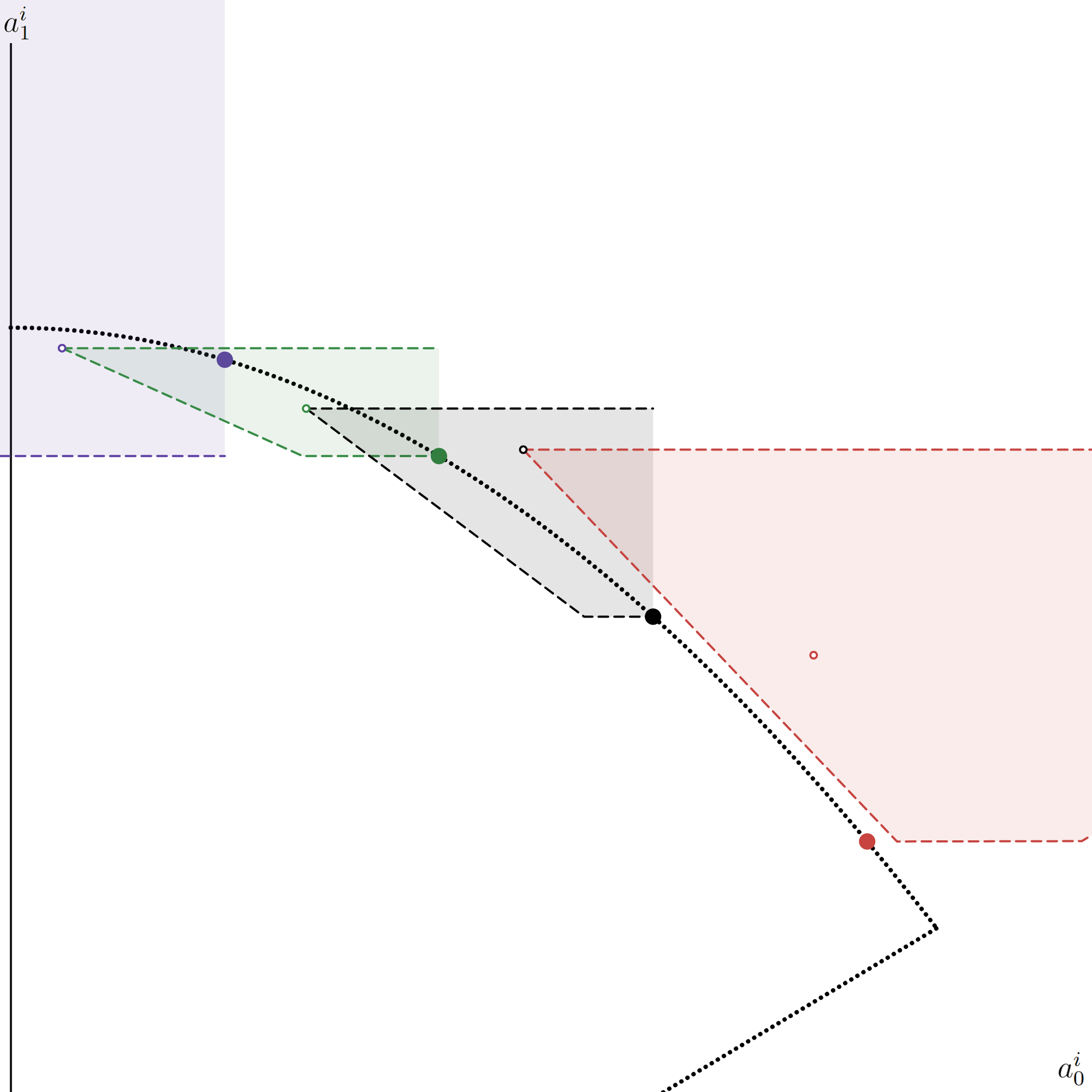}
    \caption{Sufficient Conditions with Expected Utility}
    \label{fig1}
\end{figure}

\begin{figure}
    \centering
    \includegraphics[width=1\linewidth]{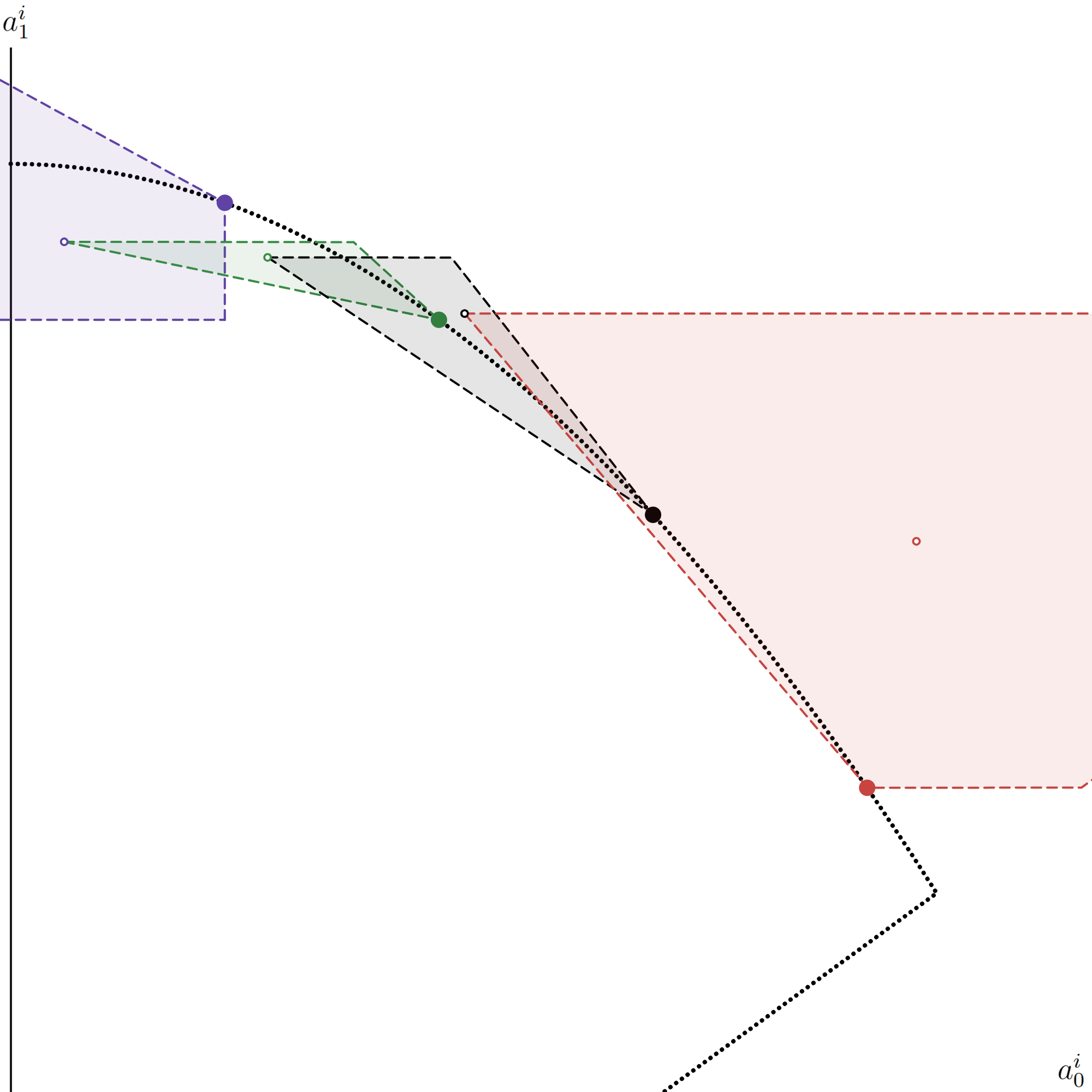}
    \caption{Sufficient Conditions with Regular Preferences}
    \label{fig2}
\end{figure}

\section{Binary Illustrations}

Before concluding the paper with applications, let us briefly scrutinize two binary environments; one in which there are two states and multiple actions, and one in which the state space is the unit interval but there are only two actions.

\subsection{Two States}\label{sectwostates}

Let \(\Theta = \left\{0,1\right\}\) and enumerate the set of actions \(A = \left\{a^1, \dots, a^m\right\} \subset \mathbb{R}\), with \(a^1 \triangleleft \cdots \triangleleft a^m\), 
so that \(a^i_0\) and \(a^i_1\) satisfy
\[a^1_{0} > \dots > a^m_{0} \quad \text{and} \quad a^1_{1} < \dots < a^m_{1}\text{.}\]
Recall also that for each \(i \in \left\{1,\dots,m\right\}\), \(a^{i}_0 \leq a^{i}_1\), which means that each point \(\left(a^1_{0}, a^1_{1}\right)\) lies on family of strictly decreasing curves in the first quadrant. We specify further that there is a curve on which these points lie that is strictly concave, which is equivalent to no action being weakly dominated for a risk-neutral agent.

The non-EU environment is simpler at first glance: the concavity of the curve on which the points lie ensure that the sufficiency conditions with respect to adjacent points imply the global conditions. However, non-adjacent actions do not affect each other in the EU setting as well.\footnote{Please visit Appendix \ref{binaryenviron} for a formal result and its proof.}

%Not so in the EU environment: non-adjacent points may matter. 
Figures \ref{fig1} and \ref{fig2} depict a four-action scenario. In both, the four solid points, in red, black, green, and purple, are the vectors of payoffs to the actions \(a^1\), \(a^2\), \(a^3\), and \(a^4\); where the \(x\)-coordinate is the monetary payoff to the action in state \(0\) and the \(y\)-coordinate is the payoff in state \(1\). Figure \ref{fig2} illustrates the environment of \(\S\)\ref{beyondeu}, whereas Figure \ref{fig1} concerns the EU setting.

In each figure, the sufficiency conditions of Corollaries \ref{corr1} and \ref{corr43} correspond to the colored regions: each point can be moved anywhere within the shaded region of the same color, with the hollowed\(+\)haloed points being specific new points. The shaded regions depend on where the new points are. Try it yourself (\href{https://www.desmos.com/calculator/sn0flbo8fb}{Link to EU example} and \href{https://www.desmos.com/calculator/qhnhhqpc6m}{Link to regular example}) by moving the points around!

\subsection{Two Actions}

\begin{figure}
    \centering
    \includegraphics[width=1\linewidth]{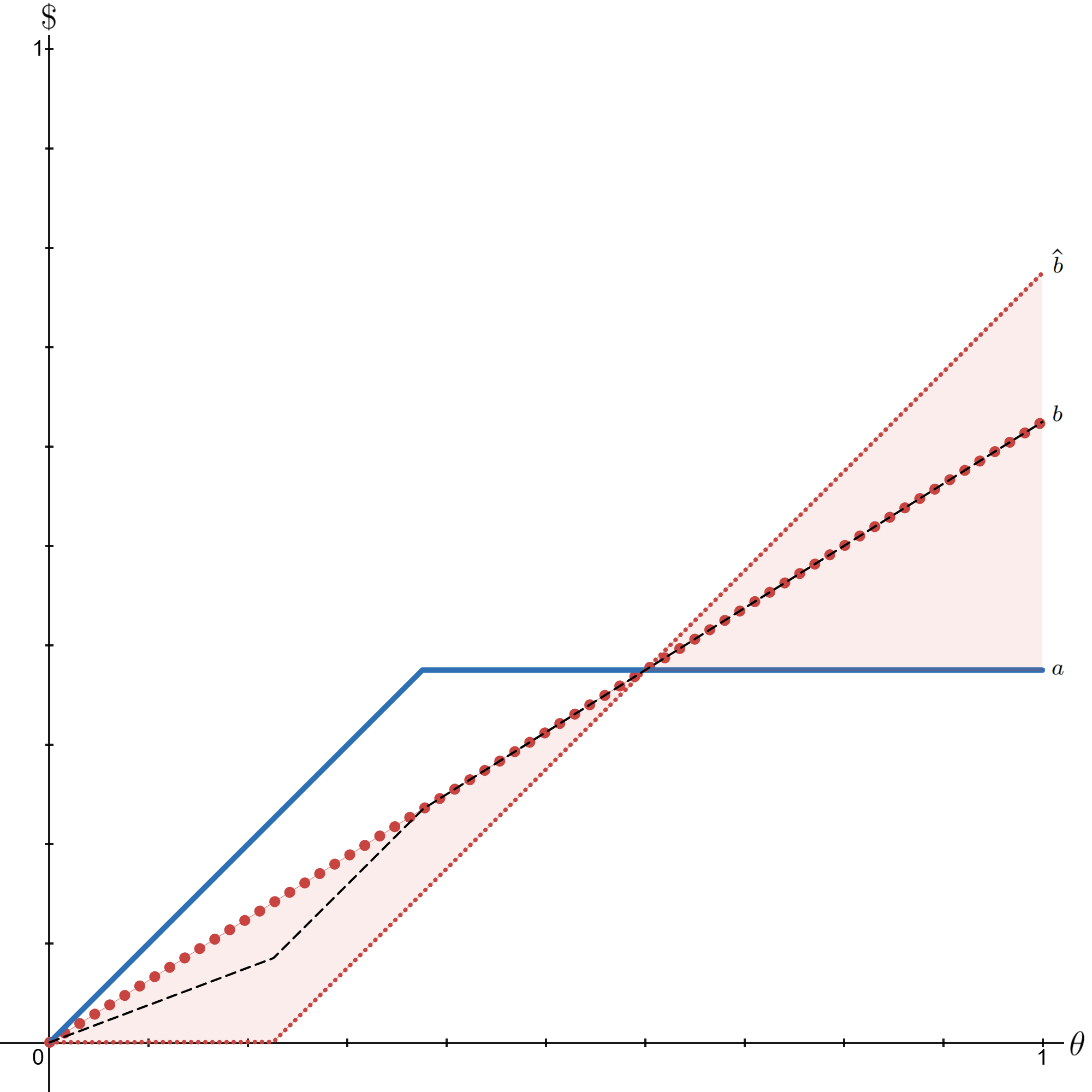}
    \caption{Reducing the DM's Security Choice}
    \label{fig3}
\end{figure}

Now let \(\Theta = \left[0,1\right]\) and \(A = \left\{a,b\right\}\), with \(b \triangleright a\). The monetary payoffs to actions \(a\) and \(b\) are the increasing functions (in \(\theta\)) \(a,b \colon \left[0,1\right] \to \mathbb{R}\). For simplicity, we assume that there exists some unique indifference state \(\theta^{\dagger}\) such that \(a_\theta > \ (<) \ b_\theta\) for all \(\theta < \ (>) \ \theta^{\dagger}\). Note that with some additional extra structure--refer to, e.g., \cite*{nachman1994optimal} and \cite*{demarzo2005bidding}--\(a\) and \(b\) can be understood as securities, with the state being an asset's random cash flow.\footnote{Securities do not quite fit the assumptions of this paper as they do not satisfy single-crossing differences, the monetary rewards being equal at \(0\). However, this is inconsequential; we can still appeal to our earlier results for securities that satisfy single-crossing differences on \(\left(0,1\right]\).} %We make another simplifying assumption, which is that \(f_a\) and \(f_b\) are continuous on \(\left[0,1\right]\).

Suppose the payoffs to \(a\) and \(b\) are transformed. When will the lower action, \(a\), become no less attractive as a result? Our first observation is that if \(a\) and \(b\) are continuous at \(\theta^{\dagger}\) pre- and post-transformation, if \(a\) becomes no less attractive (the transformation reduces the DM's action), it must be that \(a_{\theta^{\dagger}} = b_{\theta^{\dagger}} = \hat{a}_{\theta^{\dagger}} = \hat{b}_{\theta^{\dagger}}\).
This is because the order-preserving property of the transformation means that \(\hat{a}_{\theta^{\dagger}} = \hat{b}_{\theta^{\dagger}}\); so continuity tells us that if they are shifted up or down, \(a\) or \(b\) are not made steeper, which is necessary for a reduction in the DM's action.

Second, if the transformation reduces the DM's action, Proposition \ref{binaryprop} also tells us that two pointwise-dominance shifts must occur: \(\hat{b}_\theta \leq b_\theta\) for all \(\theta < \theta^{\dagger}\) and \(\hat{a}_{\theta'} \geq a_{\theta'}\) for all \(\theta' > \theta^{\dagger}\). Third, a super-actuarial improvement is, in this context,
\[\frac{\min\left\{\hat{a}_{\theta}, a_{\theta}\right\} - \hat{b}_{\theta}}{a_{\theta} - b_{\theta}} \geq \frac{\hat{b}_{\theta'} - \hat{a}_{\theta'}}{\min\left\{\hat{b}_{\theta'}, b_{\theta'}\right\}-a_{\theta'}}\text{,}\]
for all \(\theta' > \theta^{\dagger} > \theta\). Of course, these are also sufficient for the transformation to reduce the DM's action.

Figure \ref{fig3} depicts the case where the actions are securities, with \(a\) being debt and \(b\) equity. The corresponding curves pre-transformation, \(a_\theta\) and \(b_\theta\),\footnote{This is an abuse of notation, we mean the curves \(a\colon\Theta \to \mathbb{R}\) and \(b\colon\Theta \to \mathbb{R}\)--recall \(a_\theta \equiv a\left(\theta\right)\).} are in solid blue and big red dots. In a sense, for the transformation to reduce the DM's action, \(b\) must become more like a call option. Specifically, with the additional structure (standard in the security-design literature) of joint limited liability and a slope of lying in the unit interval, the shaded red region is the area in which \(\hat{b}_\theta\) (the curve corresponding to \(b\) post-transformation) must lie. 

The curve in small red dots is the case in which \(b\) is a call option following the transformation. \(b_\theta\) lies weakly above a convex combination of \(a_\theta\) and \(\hat{b}_{\theta}\), and so the first half of Definition \ref{mwcs} is satisfied. This convex combination is the dashed black line. Thus, if the second half of Definition \ref{mwcs} holds as well--\(\hat{a}_\theta\) lies weakly above some convex combination of \(a_\theta\) and \(\hat{b}_{\theta}\)--the transformation reduces the DM's action for any DM with strongly convex preferences over actions. Notably, we see that if a DM prefers debt to equity, she prefers equity to the call option. More generally, transformations that lower the DM's action are those in which the securities become more convex in a vague sense.

Alternatively, taking the perspective of Corollary \ref{corrchomp}, we see that an investor preferring debt to equity tells us more about her risk-preferences and beliefs than her preferring equity to a call option. Eschewing the steeper option tells us less--has less predictive power--the steeper a menu is.

\section{Further Applications and Discussion}

\subsection{Investing in a Risky Asset}\label{invest} An investor with wealth normalized to \(0\) is choosing how much to allocate between a riskless asset with a deterministic gross return also normalized to \(0\) and a random gross return of \(R\), which is, therefore, a real-valued random variable. We further assume that it is supported on a closed interval in whose interior \(0\) lies. 

The investor allocates \(\rho \in \mathbb{R}\) to the risky asset, which yields a stochastic terminal wealth of \(\rho R\). Naturally, the realized value of \(R\), \(r\), is the state. By assumption, \(r\) takes strictly positive and negative values and so no action dominates another. It is easy to see that this decision problem is monotone.

Now let us modify the payoff to the risky asset by assuming that there is a strictly increasing, bounded function \(\sigma \colon \mathbb{R} \to \mathbb{R}\) with \(\sigma (0) = 0\), \(\sigma(r) < r\) for all \(r < 0\) and \(\sigma(r) > 0\) for all \(r > 0\). Letting \(r_+\) and \(r_{-}\) denote strictly positive and negative states, we see that Inequality \ref{in1} is equivalent to, for all \(r_{+}\) and \(r_{-}\) and \(\rho_2 > \rho_1\),
\[\frac{\rho_1 \sigma(r_{-}) - \rho_2 \sigma(r_{-})}{\rho_1 r_{-} - \rho_2 r_{-}} \geq \frac{\rho_2 \sigma(r_{+}) - \rho_1 \sigma(r_{+})}{\rho_2 r_{+} - \rho_1 r_{+}} \quad \Leftrightarrow \quad \frac{\sigma(r_{-})}{r_{-}} \geq \frac{\sigma(r_{+})}{r_{+}}\text{.}\]
Then,
\begin{remark}
    For an SEU DM, demand in the risky asset is not increasing if \(\frac{\sigma(r_{-})}{r_{-}} \geq \frac{\sigma(r_{+})}{r_{+}}\) for all \(r_{+} > 0\) and \(r_{-} < 0\).
\end{remark}
One example of such a transformation is for all positive realizations of \(R\) to be scaled by a common \(k > 1\) and all negative realizations by a common \(s > k\).

\subsection{Insurance}

A DM with wealth normalized to \(0\) faces a potential loss of \(L > 0\). She can choose any level of insurance coverage \(\iota\) from a discrete subset of \(\left[0,L\right]\), containing both \(0\) and \(L\), at a price \( p \in \left(0,1\right)\). Let \(\Theta = \left\{0,1\right\}\) be the state space, with \(1\) the ``high'' (no loss) state. Thus, \[a^{\iota}_1 = - p \iota, \quad \text{and} \quad a^{\iota}_0 = - L + \iota (1-p)\text{.}\]

Let us modify the price to \(\hat{p}\) and loss to \(\hat{L} > 0\). How does this affect the DM's purchase? After some algebra, the conditions for Theorem \ref{steeptheorem} simplify to \(\hat{p} \leq p\) and \(L \leq \hat{L}\). Actuarial fairness simplifies to \(\hat{p} \leq p\) and so we conclude that
\begin{remark}
    Demand for insurance must increase if and only if both the per-unit price decreases and the loss amount increases.
\end{remark}
Notably, no amount of a price decrease can make up for a strictly decreased loss amount. Likewise, no increase in the loss amount increase can make up for a strict price increase.

\subsection{Cooperation in the Infinitely Repeated Prisoners' Dilemma}

\begin{figure}
\centering
  \begin{subfigure}{7cm}
\centering
\LARGE
\begin{game}{2}{2}[P1][P2] % A 2x2 game
    & \(c\) & \(d\) \\
\(c\) & \(1,1\) & \(-\gamma,\beta\) \\
\(d\) & \(\beta, -\gamma\) & \(0,0\) \\
\end{game}
    \caption{Pre-Transformation}
    \label{pd1}
  \end{subfigure} \begin{subfigure}{7cm}
\centering
\LARGE
\begin{game}{2}{2}[P1][P2] % A 2x2 game
    & \(c\) & \(d\) \\
\(c\) & \(\hat{\alpha},\hat{\alpha}\) & \(-\hat{\gamma},\hat{\beta}\) \\
\(d\) & \(\hat{\beta}, -\hat{\gamma}\) & \(\hat{\rho}, \hat{\rho}\) \\
\end{game}
    \caption{Post-Transformation}
    \label{pd2}
  \end{subfigure}
  \caption{The Prisoner's Dilemma}\label{pd}
\end{figure}

Our third application is superficially quite distant from the topic of this paper. Nevertheless, as we will see, the mechanics will be such that our first proposition applies. Consider the following question. Take a standard infinitely repeated prisoner's dilemma (with perfect monitoring), with known \textit{monetary} payoffs to the two players. Suppose they are known to be exponential discounters with risk-averse utilities in money. Suppose, moreover, that it is known that cooperation can be sustained in a sub-game perfect equilibrium (SPE) of the infinitely-repeated game. In which other prisoner's dilemmas--which other ordinally equivalent games--must there be an SPE of the infinitely-repeated game in which cooperation can be sustained?

The normal forms of the games, pre- and post-transformation are depicted in Figure \ref{pd}. As they are prisoner's dilemmas \(\beta > 1\), \(\hat{\beta} > \hat{\alpha} > \hat{\rho} > - \hat{\gamma}\), and \(\gamma > 0\). Always cooperating on-path is sustainable as part of an SPE if and only if it can be sustained by ``grim-trigger:'' cooperate at every history in which only \((c,c)\) has been played; else, defect. This means tha it must be the case that
\(1 \geq (1-\delta) \beta\). But \(\delta\) is a number between \(0\) and \(1\), so we may treat it like a probability. Accordingly, letting \(\delta\) denote the probability of ``state'' \(1\), \(c_0 = c_1 = 1\), \(d_0 = \beta\), and \(d_1 = 0\). \(d \triangleright c\), and Proposition \ref{binaryprop} immediately applies:
\begin{corollary}
    There must exist an SPE that sustains correlation if and only if the punishment is harsher, \(\hat{\rho} \leq 0\); the cooperation payoff is larger, \(1 \leq \hat{\alpha}\); and the immediate benefit of betrayal is relatively smaller: either \(\hat{\beta} \leq \beta\), or 
    \[(1-\hat{\rho})(\beta-1) + \hat{\alpha} \geq \hat{\beta} > \beta\text{.}\]
\end{corollary}
Any strict increase in the payoff to \((d,d)\) can render cooperation unsustainable, no matter how the other payoffs are altered. Likewise, the payoff to \((c,c)\) cannot decrease if cooperation must survive. The immediate impact of betrayal is more subtle: it can either decrease (\(\hat{\beta} \leq \beta\)), which is obviously good for cooperation; or increase, in which case it cannot increase too much relative to the new punishment payoff \(\hat{\rho}\) and cooperation payoff \(\hat{\alpha}\).

\subsection{Introducing a Lower Bound on Risk Aversion}\label{lowerbound}

Suppose we discover that the DM is at least as risk averse as an agent with utility function \(v\) in the sense that the DM's utility is a strictly increasing and weakly concave transformation of \(v\). That is, we are made aware of a lower bound on the DM's risk aversion. When does this lower the action we know the DM takes versus the case in which our lower bound is risk-neutrality?

Suppose \(v\) is such that for all \(a \in A\) and \(b \in A_{\triangleright}^a\) and for all \(\theta \in \mathcal{A}\) and \(\theta' \in \mathcal{B}\), \(v(b_\theta) \leq b_{\theta}\), \(v(a_\theta) \leq a_\theta\), \(v(b_{\theta'}) \geq b_{\theta'}\), \(v(a_{\theta'}) \geq a_{\theta'}\). In this case, Inequality \ref{in1} becomes
\[\frac{v(a_\theta)-v(b_\theta)}{a_\theta - b_\theta} \geq \frac{v(b_{\theta'})-v(a_{\theta'})}{b_{\theta'} - a_{\theta'}}\text{,}\]
which must hold due to the concavity of \(v\) and the single-crossing differences assumption.

The point of all of this is to illustrate that under this variety of transformation--introducing a lower bound on risk aversion--our sufficient condition simplifies to a collection of ordinal conditions on the relationship between \(v(\cdot)\) and \(\cdot\); we get the cardinal actuarial-improvement inequalities for free. It is also notable that an increase in the lower bound on risk aversion need not lead to lower actions. To wit, if \(v(\cdot) > \cdot\) for all \(\cdot\) or  \(v(\cdot) < \cdot\), then irrespective of \(v\)'s shape, the necessary condition in Theorem \ref{steeptheorem} will not hold.

\subsection{Ending on a Negative Note}

Now let us continue with the theme we ended \(\S\)\ref{lowerbound} on. Namely, the necessary condition identified in Theorem \ref{steeptheorem} is quite strong and fails in many environments for common modifications. In particular, any change that increases--in turn, lowers--payoffs to every action in every state, violates relevant steepness. Thus, such a change does not reduce the DM's action, i.e., may result in an increased action. 

These changes are ubiquitous. They include, for instance, the addition or subtraction of a constant to the actions' monetary rewards. Consequently, changes to the DM's wealth do not reduce her action. Another example is \(\S\)\ref{invest}'s investment problem. There, we need the change in the return to be vaguely ``S-shaped,'' with an inflection point precisely at the risk-free rate. A global shift up (or down) in the return would not reduce the DM's action. A classic comparative statics environment is that of a monopolist choosing how much quantity to produce: its profit is \(\Pi(q) = q P_\theta(q) - c_\theta (q)\), where \(P_\theta\) is inverse demand and \(c_\theta\) is its cost. Theorem \ref{steeptheorem} indicates that a pointwise shift up or down in either demand or cost does not reduce the monopolist's action. Subsidies and taxes produce such pointwise shifts, making it impossible to robustly predict pass-through.

\bibliography{scs.bib}

\appendix

\section{Omitted Proofs}

\subsection{Proposition \ref{binaryprop} Proof}\label{binarypropproof}
\begin{proof}
    Let \(b \triangleright a\). A transformation reducing the DM's action is equivalent to the set of beliefs at which the DM prefers \(a\) to \(b\) increasing in size--in a set-inclusion sense--for any \(u \in \mathcal{U}\). Define \[\mathcal{C} \coloneqq \left\{\theta \in \Theta \colon a_{\theta} = b_{\theta}\right\} \text{,}\]
    and recall that \(\mathcal{A}\) and \(\mathcal{B}\) are the sets of states for which the DM strictly prefers \(a\) to \(b\) and vice-versa. Given \(u \in \mathcal{U}\) and for fixed pair of states \(\theta \in \mathcal{A}\) and \(\theta' \in \mathcal{B}\), let \(\mu_{u}^{\theta, \theta'}\) and \(\hat{\mu}_{u}^{\theta, \theta'}\) be the solutions (in \(\mu\)) to
    \[\mu u(a_{\theta'}) + (1-\mu) u(a_{\theta}) = \mu u(b_{\theta'}) + (1-\mu) u(b_{\theta})\text{,}\]
    and
    \[\mu u(\hat{a}_{\theta'}) + (1-\mu) u(\hat{a}_{\theta}) = \mu u(\hat{b}_{\theta'}) + (1-\mu) u(\hat{b}_{\theta})\text{,}\]
    respectively. 

    For any \(\theta \in \Theta\), let \(v_{\theta}\) denote the corresponding vertex (as a point in the simplex). Then, pre-transformation, the extreme points of the set of beliefs at which the DM prefers \(a\) to \(b\) is \[\left\{v_\theta \colon \theta \in \mathcal{A} \cup \mathcal{C}\right\} \cup \left\{\mu_{u}^{\theta, \theta'} \colon \theta \in \mathcal{A}, \ \theta' \in \mathcal{B}\right\}\text{,}\] i.e., the set of vertices at which \(a\) is uniquely optimal, and the points on the edges connecting the vertices at which \(a\) is uniquely optimal with the vertices at which \(b\) is uniquely optimal. The analogous set post-transformation is  \[\left\{v_\theta \colon \theta \in \mathcal{A} \cup \mathcal{C}\right\} \cup \left\{\hat{\mu}_{u}^{\theta, \theta'} \colon \theta \in \mathcal{A}, \ \theta' \in \mathcal{B}\right\}\text{.}\]
    By the Krein-Milman theorem, the sets of beliefs at which the DM prefers \(a\) to \(b\) pre- and post-transformation are the closed convex hulls of their extreme points. Accordingly, a transformation reduces the DM's action if and only if for any \(u\) and pair \(\left(\theta,\theta'\right) \in \mathcal{A} \times \mathcal{B}\), \(\mu_{u}^{\theta, \theta'} \leq \hat{\mu}_{u}^{\theta, \theta'}\).

    \medskip

    \noindent \(\left(\Leftarrow\right)\) Fix an arbitrary pair of states \(\theta \in \mathcal{A}\) and \(\theta' \in \mathcal{B}\) and an arbitrary (risk-averse) utility \(u \in \mathcal{U}\), and let \(\hat{b}_{\theta} \leq b_{\theta}\), \(a_{\theta'} \leq \hat{a}_{\theta'}\), and \(a\) be a super-actuarial improvement versus \(b\). Note that this is equivalent to
    \[\label{ina1}\tag{\(A1\)}\left(\frac{a_\theta - b_\theta}{\min\left\{\hat{a}_\theta, a_\theta\right\} - \hat{b}_\theta}\right) \frac{\hat{b}_{\theta'} - \hat{a}_{\theta'}}{\min\left\{\hat{b}_{\theta'}, b_{\theta'}\right\}-a_{\theta'}} \leq 1 \text{.}\]
     Define \(\underline{a}_{\theta} \coloneqq \min\left\{\hat{a}_\theta, a_\theta\right\}\) and \(\underline{b}_{\theta'} \coloneqq \min\left\{\hat{b}_{\theta'}, b_{\theta'}\right\}\). Then,
\[\begin{split}
        \hat{\mu}_{u}^{\theta, \theta'} = \frac{u\left(\hat{a}_{\theta}\right) - u\left(\hat{b}_{\theta}\right)}{u\left(\hat{a}_{\theta}\right) - u\left(\hat{b}_{\theta}\right) + u\left(\hat{b}_{\theta'}\right) - u\left(\hat{a}_{\theta'}\right)} &\geq \frac{u\left(\underline{a}_{\theta}\right) - u\left(\hat{b}_{\theta}\right)}{u\left(\underline{a}_{\theta}\right) - u\left(\hat{b}_{\theta}\right) + u\left(\hat{b}_{\theta'}\right) - u\left(\hat{a}_{\theta'}\right)} \\ &= \frac{\frac{u\left(\underline{a}_{\theta}\right) - u\left(\hat{b}_{\theta}\right)}{\underline{a}_{\theta} - \hat{b}_{\theta}}}{\frac{u\left(\underline{a}_{\theta}\right) - u\left(\hat{b}_{\theta}\right)}{\underline{a}_{\theta} - \hat{b}_{\theta}} + \frac{u\left(\hat{b}_{\theta'}\right) - u\left(\hat{a}_{\theta'}\right)}{\underline{a}_{\theta} - \hat{b}_{\theta}}}\\
        &\geq \frac{\frac{u\left(a_{\theta}\right) - u\left(b_{\theta}\right)}{a_{\theta} - b_{\theta}}}{\frac{u\left(a_{\theta}\right) - u\left(b_{\theta}\right)}{a_{\theta} - b_{\theta}} + \frac{u\left(\hat{b}_{\theta'}\right) - u\left(\hat{a}_{\theta'}\right)}{\underline{a}_{\theta} - \hat{b}_{\theta}}}\text{,}
    \end{split}\]
    where the first inequality follows from the definition of \(\underline{a}_\theta\) and the second from the Three-chord lemma (Theorem 1.16 in \cite*{phelps2009convex}).

    The last expression on the right-hand side can be rewritten as
    \[\tag{\(A2\)}\label{exa2} \frac{u\left(a_{\theta}\right) - u\left(b_{\theta}\right)}{u\left(a_{\theta}\right) - u\left(b_{\theta}\right) + \left(\frac{a_{\theta} - b_{\theta}}{\underline{a}_{\theta} - \hat{b}_{\theta}}\right)\left(\hat{b}_{\theta'} - \hat{a}_{\theta'}\right)\frac{u\left(\hat{b}_{\theta'}\right) - u\left(\hat{a}_{\theta'}\right)}{\hat{b}_{\theta'} - \hat{a}_{\theta'}}}\text{.}\]
    If \(\hat{b}_{\theta'} \geq b_{\theta'}\), then by the Three-chord lemma, Expression \ref{exa2} is weakly greater than
    \[\frac{u\left(a_{\theta}\right) - u\left(b_{\theta}\right)}{u\left(a_{\theta}\right) - u\left(b_{\theta}\right) + \left(\frac{a_{\theta} - b_{\theta}}{\underline{a}_{\theta} - \hat{b}_{\theta}}\right)\left(\frac{\hat{b}_{\theta'} - \hat{a}_{\theta'}}{b_{\theta'} - a_{\theta'}}\right)\left(u\left(b_{\theta'}\right) - u\left(a_{\theta'}\right)\right)}\text{,}\]
    which is weakly greater than
    \[\frac{u\left(a_{\theta}\right) - u\left(b_{\theta}\right)}{u\left(a_{\theta}\right) - u\left(b_{\theta}\right) + u\left(b_{\theta'}\right) - u\left(a_{\theta'}\right)} = \mu_{u}^{\theta, \theta'}\text{,}\]
    by Inequality \ref{ina1}.

    If \(\hat{b}_{\theta'} < b_{\theta'}\), then by the Three-chord lemma, Expression \ref{exa2} is weakly greater than
    \[\frac{u\left(a_{\theta}\right) - u\left(b_{\theta}\right)}{u\left(a_{\theta}\right) - u\left(b_{\theta}\right) + \left(\frac{a_{\theta} - b_{\theta}}{\underline{a}_{\theta} - \hat{b}_{\theta}}\right)\left(\frac{\hat{b}_{\theta'} - \hat{a}_{\theta'}}{\hat{b}_{\theta'} - a_{\theta'}}\right)\left(u\left(\hat{b}_{\theta'}\right) - u\left(a_{\theta'}\right)\right)}\text{,}\]
    which Inequality \ref{ina1} implies is weakly greater than
    \[\frac{u\left(a_{\theta}\right) - u\left(b_{\theta}\right)}{u\left(a_{\theta}\right) - u\left(b_{\theta}\right) + u\left(\hat{b}_{\theta'}\right) - u\left(a_{\theta'}\right)} \geq \frac{u\left(a_{\theta}\right) - u\left(b_{\theta}\right)}{u\left(a_{\theta}\right) - u\left(b_{\theta}\right) + u\left(b_{\theta'}\right) - u\left(a_{\theta'}\right)} = \mu_{u}^{\theta, \theta'}\text{,}\]
    where the inequality follows from the monotonicity of \(u\). 

    \medskip

    \noindent \(\left(\Rightarrow\right)\) Fix an arbitrary pair of states \(\theta \in \mathcal{A}\) and \(\theta' \in \mathcal{B}\). Recall the ``actuarial fairness inequality,'' Inequality \ref{in1app}:
    \[\label{in1app}\tag{\(1\)}\frac{\min\left\{\hat{a}_\theta, a_\theta\right\} - \hat{b}_\theta}{a_\theta - b_\theta} \geq \frac{\hat{b}_{\theta'} - \hat{a}_{\theta'}}{\min\left\{\hat{b}_{\theta'}, b_{\theta'}\right\}-a_{\theta'}}\text{.}\]
    There are six cases to consider, proved, with the exception of the fourth, by contraposition.
    \begin{enumerate}
        \item \(\hat{b}_\theta > b_\theta\);
        \item \(\hat{a}_{\theta'} < a_{\theta'}\);
        \item \(a_\theta \geq \hat{a}_{\theta}\) and \(\hat{b}_{\theta'} \geq b_{\theta'}\), but Inequality \ref{in1app} does not hold;
        \item \(\hat{a}_\theta > a_{\theta} > b_\theta \geq \hat{b}_\theta\) and \(b_{\theta'} > \hat{b}_{\theta'} > \hat{a}_{\theta'} \geq a_{\theta'}\);
        \item \(\hat{a}_\theta > a_{\theta} > b_\theta \geq \hat{b}_\theta\) and \(\hat{b}_{\theta'} \geq b_{\theta'}\), but Inequality \ref{in1app} does not hold; and
        \item \(a_{\theta} \geq \hat{a}_\theta\) and \(b_{\theta'} > \hat{b}_{\theta'} > \hat{a}_{\theta'} \geq a_{\theta'}\) but Inequality \ref{in1app} does not hold.
    \end{enumerate}
    \begin{claim}
        The cases overlap but are exhaustive.
    \end{claim}
    \begin{proof}
        The first two establish the necessity of \(\hat{b}_\theta \leq b_\theta\) and \(\hat{a}_{\theta'} \geq a_{\theta'}\). There are four possible arrangements of \(\hat{a}_\theta\) and \(\hat{b}_{\theta'}\): 
        \[a_\theta \underset{\textcolor{OrangeRed}{<}}{\geq} \hat{a}_\theta, \quad \text{and} \quad \hat{b}_{\theta'} \underset{\textcolor{OrangeRed}{<}}{\geq} b_{\theta'}\]
        (black/black, red/red, red/black, and black/red), which are cases 3, 4, 5, and 6, respectively.\end{proof}
    
    \noindent \textbf{Case 1.} Let \(\hat{b}_{\theta} > b_{\theta}\). Let
\[u(x) = \begin{cases}
    x, \quad &\text{if} \quad x \leq \min\left\{a_{\theta}, \hat{b}_{\theta}\right\}\\
    \iota x + (1-\iota) \min\left\{a_{\theta}, \hat{b}_{\theta}\right\}, \quad &\text{if} \quad \min\left\{a_{\theta}, \hat{b}_{\theta}\right\} < x\text{,}
\end{cases}\]
for some \(\iota \in \left(0,1\right)\). Then,
\[\mu_u = \frac{\iota a_\theta + (1-\iota) \min\left\{a_{\theta}, \hat{b}_{\theta}\right\} - b_{\theta}}{\iota a_\theta + (1-\iota) \min\left\{a_{\theta}, \hat{b}_{\theta}\right\} - b_{\theta} + \iota\left(b_{\theta'}-a_{\theta'}\right)}\text{,}\]
and
\[\hat{\mu}_u = \frac{\iota \left(\hat{a}_\theta -\hat{b}_{\theta}\right)}{\iota \left(\hat{a}_\theta -\hat{b}_{\theta}\right) + \iota \left(\hat{b}_{\theta'} -\hat{a}_{\theta'}\right)} = \frac{\hat{a}_\theta -\hat{b}_{\theta}}{\hat{a}_\theta -\hat{b}_{\theta} + \hat{b}_{\theta'} -\hat{a}_{\theta'}}\text{.}\]
By construction, for all sufficiently small \(\iota > 0\), \(\mu_u  - \hat{\mu}_u > 0\).

\bigskip

\noindent \textbf{Case 2.} Let \(\hat{a}_{\theta'} < a_{\theta'}\). Let 
\[u(x) = \begin{cases}
    x, \quad &\text{if} \quad x \leq a_{\theta'}\\
    \iota x + (1-\iota) a_{\theta'}, \quad &\text{if} \quad a_{\theta'} < x\text{,}
\end{cases}\]
for some \(\iota \in \left(0,1\right)\). Then,
\[\mu_u = \frac{a_\theta - b_{\theta}}{a_\theta - b_{\theta} + \iota\left(b_{\theta'}-a_{\theta'}\right)}\text{,}\]
and
\[\hat{\mu}_u = \frac{\hat{a}_\theta -\hat{b}_{\theta}}{\hat{a}_\theta -\hat{b}_{\theta} + u\left(\hat{b}_{\theta'}\right) -\hat{a}_{\theta'}} \leq \frac{\hat{a}_\theta -\hat{b}_{\theta}}{\hat{a}_\theta -\hat{b}_{\theta} + \min\left\{a_{\theta'}, \hat{b}_{\theta'}\right\} -\hat{a}_{\theta'}} < 1\text{.}\]
By construction, for all sufficiently small \(\iota > 0\), \(\mu_u  - \hat{\mu}_u > 0\).

\bigskip

\noindent \textbf{Case 3.} If \(a_\theta \geq \hat{a}_{\theta}\) and \(\hat{b}_{\theta'} \geq b_{\theta'}\), Inequality \ref{in1app} simplifies to
\[\frac{\hat{a}_\theta - \hat{b}_\theta}{a_\theta - b_\theta} \geq \frac{\hat{b}_{\theta'} - \hat{a}_{\theta'}}{b_{\theta'}-a_{\theta'}}\text{.}\]
If this does not hold, \(\mu_u  - \hat{\mu}_u > 0\) for a risk-neutral DM.

\bigskip

\noindent \textbf{Case 4.} Let \(\hat{a}_\theta > a_{\theta} > b_\theta \geq \hat{b}_\theta\) and \(b_{\theta'} > \hat{b}_{\theta'} > \hat{a}_{\theta'} \geq a_{\theta'}\). Then, observe that
\[\frac{\underline{a}_{\theta} - \hat{b}_{\theta}}{a_\theta - b_\theta} = \frac{a_\theta - \hat{b}_{\theta}}{a_\theta - b_\theta} \geq 1 \geq \frac{\hat{b}_{\theta'} - \hat{a}_{\theta'}}{\hat{b}_{\theta'} - a_{\theta'}} = \frac{\hat{b}_{\theta'} - \hat{a}_{\theta'}}{\underline{b}_{\theta'} - a_{\theta'}}\text{,}\]
so Inequality \ref{in1app} must hold.

\bigskip

\noindent \textbf{Case 5.} Let \(\hat{a}_\theta > a_{\theta} > b_\theta \geq \hat{b}_\theta\) and \(\hat{b}_{\theta'} \geq b_{\theta'}\), but Inequality \ref{in1app} does not hold. Then, let
\[u(x) = \begin{cases}
    x, \quad &\text{if} \quad x \leq a_{\theta}\\
    \iota x + (1-\iota) a_{\theta}, \quad &\text{if} \quad a_{\theta} < x\text{,}
\end{cases}\]
for some \(\iota \in \left(0,1\right)\). Then,
\[\mu_u - \hat{\mu}_u = \frac{a_\theta - b_{\theta}}{a_\theta - b_{\theta} + \iota\left(b_{\theta'}-a_{\theta'}\right)} - \frac{\iota \hat{a}_\theta + (1-\iota) a_\theta -\hat{b}_{\theta}}{\iota \hat{a}_\theta + (1-\iota) a_\theta -\hat{b}_{\theta} + \iota\left(\hat{b}_{\theta'} -\hat{a}_{\theta'}\right)}\]
has the same sign as \[\frac{\hat{b}_{\theta'} - \hat{a}_{\theta'}}{b_{\theta'} - a_{\theta'}} - \frac{\iota \hat{a}_{\theta} + \left(1-\iota\right)a_{\theta} - \hat{b}_{\theta}}{a_{\theta} - b_{\theta}}\text{,}\]
which equals
\[\frac{\hat{b}_{\theta'} - \hat{a}_{\theta'}}{\underline{b}_{\theta'} - a_{\theta'}} - \frac{\iota \hat{a}_{\theta} + \left(1-\iota\right)\underline{a}_{\theta} - \hat{b}_{\theta}}{a_{\theta} - b_{\theta}}\text{.}\]
As Inequality \ref{in1app} does not hold, this expression is strictly positive for all sufficiently small \(\iota > 0\).

\bigskip

\noindent \textbf{Case 6.} Finally, let \(a_{\theta} \geq \hat{a}_\theta\) and \(b_{\theta'} > \hat{b}_{\theta'} > \hat{a}_{\theta'} \geq a_{\theta'}\) but Inequality \ref{in1app} does not hold. Let
\[u(x) = \begin{cases}
    x, \quad &\text{if} \quad x \leq \hat{b}_{\theta'}\\
    \iota x + (1-\iota) \hat{b}_{\theta'}, \quad &\text{if} \quad \hat{b}_{\theta'} < x\text{,}
\end{cases}\]
for some \(\iota \in \left(0,1\right)\). Then,
\[\mu_u - \hat{\mu}_u = \frac{a_\theta - b_{\theta}}{a_\theta - b_{\theta} + \iota b_{\theta'} + \left(1-\iota\right) \hat{b}_{\theta'} - a_{\theta'}} - \frac{\hat{a}_\theta -\hat{b}_{\theta}}{\hat{a}_\theta -\hat{b}_{\theta} + \hat{b}_{\theta'} - \hat{a}_{\theta'}}\]
has the same sign as \[\frac{\hat{b}_{\theta'} - \hat{a}_{\theta'}}{\iota b_{\theta'} + \left(1-\iota\right) \hat{b}_{\theta'} - a_{\theta'}} - \frac{\hat{a}_{\theta} - \hat{b}_{\theta}}{a_{\theta} - b_{\theta}}\text{,}\]
which equals
\[\frac{\hat{b}_{\theta'} - \hat{a}_{\theta'}}{\iota b_{\theta'} + \left(1-\iota\right) \underline{b}_{\theta'} - a_{\theta'}} - \frac{\underline{a}_{\theta} - \hat{b}_{\theta}}{a_{\theta} - b_{\theta}}\text{.}\]
As Inequality \ref{in1app} does not hold, this expression is strictly positive for all sufficiently small \(\iota > 0\). \end{proof}

\subsection{Theorem \ref{steeptheorem} Proof}\label{steepproof}
\begin{proof}
    First, we argue that \(\hat{a}^i_\theta \leq a^i_\theta\) for all \(a^i \in A_{\theta,\theta'} \setminus \left\{\underline{a}^{\left(\theta,\theta'\right)}\right\}\), for all \(\theta, \theta' \in \Theta\), is necessary for a transformation to reduce the DM's action. Suppose for the sake of contraposition that there exists some pair \(\theta, \theta' \in \Theta\) and action \(a^{i+1} \in A_{\theta,\theta'} \setminus \left\{\underline{a}^{\left(\theta,\theta'\right)}\right\}\) such that \(a_{\theta}^{i+1} < \hat{a}_{\theta}^{i+1}\).
    Let \[u(x) = \min\left\{x,\iota x + (1-\iota) \min\left\{a_{\theta}^{i}, \hat{a}_{\theta}^{i+1}\right\}\right\}\text{,}\]
    for some \(\iota \in \left(0,1\right)\). 
    
    Crucially, the structure placed on the decision problem and the restricted class of transformations--recall that they are ordinal-ranking preserving, across both states and actions--means that
    \begin{enumerate}
        \item\label{observation1} \(\hat{a}^j_{\theta} > \min\left\{a_{\theta}^{i}, \hat{a}_{\theta}^{i+1}\right\}\) and \(a^j_{\theta} > \min\left\{a_{\theta}^{i}, \hat{a}_{\theta}^{i+1}\right\}\) for all \(j < i\); 
        \item\label{observation2} \(\hat{a}^i_{\theta} > \hat{a}^{i+1}_{\theta} \geq \min\left\{a_{\theta}^{i}, \hat{a}_{\theta}^{i+1}\right\}\) and \(a^i_{\theta} \geq \min\left\{a_{\theta}^{i}, \hat{a}_{\theta}^{i+1}\right\}\);
        \item\label{observation3} \(\hat{a}^k_{\theta'} > \min\left\{a_{\theta}^{i}, \hat{a}_{\theta}^{i+1}\right\}\) and \(a^k_{\theta'} \geq \min\left\{a_{\theta}^{i}, \hat{a}_{\theta}^{i+1}\right\}\) for all relevant \(k\);
        and
        \item\label{observation4} \(a_{\theta}^{s} < a_{\theta}^{i+1} < \min\left\{a_{\theta}^{i}, \hat{a}_{\theta}^{i+1}\right\}\) for all \(s > i+1\).
    \end{enumerate}
    
    For any \(j < i\), the indifference belief along the edge between \(\theta\) and \(\theta'\) between actions \(a^j\) and \(a^i\), pre-transformation, is
    \[\frac{u(a^{j}_{\theta}) - u(a^{i}_{\theta})}{u(a^{j}_{\theta}) - u(a^{i}_{\theta}) + u(a^i_{\theta'}) - u(a^{j}_{\theta'})} = \frac{a^{j}_{\theta} - a^{i}_{\theta}}{a^{j}_{\theta} - a^{i}_{\theta} + a^i_{\theta'} - a^{j}_{\theta'}}\text{,}\]
    and the indifference belief along the edge between \(\theta\) and \(\theta'\) between actions \(a^j\) and \(a^i\), post-transformation, is
    \[\frac{u(\hat{a}^{j}_{\theta}) - u(\hat{a}^{i}_{\theta})}{u(\hat{a}^{j}_{\theta}) - u(\hat{a}^{i}_{\theta}) + u(\hat{a}^i_{\theta'}) - u(\hat{a}^{j}_{\theta'})} = \frac{\hat{a}^{j}_{\theta} - \hat{a}^{i}_{\theta}}{\hat{a}^{j}_{\theta} - \hat{a}^{i}_{\theta} + \hat{a}^i_{\theta'} - \hat{a}^{j}_{\theta'}}\text{,}\]
    where the two equalities hold as a result of Properties \ref{observation1} and \ref{observation3}. That is, by construction, each of the monetary payoffs lies above the inflection point in the utility function.
    
    Consequently,
    \[\tag{\(A3\)}\label{ex4}\begin{split}
        \underline{\mu} &\coloneqq \max_{j < i} \max \left\{\frac{u(a^{j}_{\theta}) - u(a^{i}_{\theta})}{u(a^{j}_{\theta}) - u(a^{i}_{\theta}) + u(a^i_{\theta'}) - u(a^{j}_{\theta'})}, \frac{u(\hat{a}^{j}_{\theta}) - u(\hat{a}^{i}_{\theta})}{u(\hat{a}^{j}_{\theta}) - u(\hat{a}^{i}_{\theta}) + u(\hat{a}^i_{\theta'}) - u(\hat{a}^{j}_{\theta'})}\right\}\\
        &= \max_{j < i} \max \left\{\frac{a^{j}_{\theta} - a^{i}_{\theta}}{a^{j}_{\theta} - a^{i}_{\theta} + a^i_{\theta'} - a^{j}_{\theta'}}, \frac{\hat{a}^{j}_{\theta} - \hat{a}^{i}_{\theta}}{\hat{a}^{j}_{\theta} - \hat{a}^{i}_{\theta} + \hat{a}^i_{\theta'} - \hat{a}^{j}_{\theta'}}\right\}
    \end{split}\]
    is strictly less than \(1\) for all \(\iota \in \left(0,1\right)\). 
    
    On the other hand, for all \(s \geq i+1\), the indifference belief along the edge between \(\theta\) and \(\theta'\) between actions \(a^s\) and \(a^i\), pre-transformation, is
    \[\frac{u(a^{i}_{\theta}) - u(a^{s}_{\theta})}{u(a^{i}_{\theta}) - u(a^{s}_{\theta}) + u(a^s_{\theta'}) - u(a^{i}_{\theta'})} = \frac{\iota a^{i}_{\theta} + (1-\iota) \min\left\{a_{\theta}^{i}, \hat{a}_{\theta}^{i+1}\right\}- a^{s}_{\theta}}{\iota a^{i}_{\theta} + (1-\iota) \min\left\{a_{\theta}^{i}, \hat{a}_{\theta}^{i+1}\right\}- a^{s}_{\theta} + \iota\left(a^{s}_{\theta'} - a^{i}_{\theta'}\right)}\text{,}\]
    where the equality is due to Properties \ref{observation2}, \ref{observation3}, and \ref{observation4}. That is, \(a_{\theta}^{i+1} < \hat{a}_{\theta}^{i+1}\) has allowed us to construct a utility function for which the monetary rewards to \(a^i\) and each higher \(a^s\) in state \(\theta\) are on opposite sides of the inflection point, pre-transformation. Accordingly, as \(\iota \downarrow 0\), 
    \[\overline{\mu} \coloneqq \min_{j > i} \frac{u(a^{i}_{\theta}) - u(a^{j}_{\theta})}{u(a^{i}_{\theta}) - u(a^{j}_{\theta}) + u(a^j_{\theta'}) - u(a^{i}_{\theta'})} = \min_{j > i} \frac{\iota a^{i}_{\theta} + (1-\iota) \min\left\{a_{\theta}^{i}, \hat{a}_{\theta}^{i+1}\right\}- a^{j}_{\theta}}{\iota a^{i}_{\theta} + (1-\iota) \min\left\{a_{\theta}^{i}, \hat{a}_{\theta}^{i+1}\right\}- a^{j}_{\theta} + \iota\left(a^{j}_{\theta'} - a^{i}_{\theta'}\right)}\]
    converges to \(1\). 
    
    Finally, using Properties \ref{observation2} and \ref{observation3}, the indifference belief between \(a^i\) and \(a^{i+1}\), post-transformation, is
    \[\hat{\mu} \coloneqq \frac{u(\hat{a}^{i}_{\theta}) - u(\hat{a}^{i+1}_{\theta})}{u(\hat{a}^{i}_{\theta}) - u(\hat{a}^{i+1}_{\theta}) + u(\hat{a}^{i+1}_{\theta'}) - u(\hat{a}^{i}_{\theta'})} = \frac{\hat{a}^{i}_{\theta} - \hat{a}^{i+1}_{\theta}}{\hat{a}^{i}_{\theta} - \hat{a}^{i+1}_{\theta} + \hat{a}^{i+1}_{\theta'} - \hat{a}^{i}_{\theta'}} < 1\text{,}\]
    for all \(\iota \in \left(0,1\right)\). 

    Restricting attention to the edge between \(\theta\) and \(\theta'\), pre-transformation, \(a^i\) is optimal for all beliefs in \(\left[\underline{\mu}, \overline{\mu}\right]\). As we have already argued, for all sufficiently low \(\iota > 0\), this interval is non-degenerate and as \(\iota \downarrow 0\), \(\overline{\mu}\uparrow 1\). Moreover, post-transformation \(a^{i}\) is superior to all lower actions for all beliefs in \(\left[\underline{\mu}, 1\right]\), which is non-degenerate. Post-transformation, \(a^{i}\) is strictly inferior to \(a^{i+1}\) for all beliefs in \(\left(\hat{\mu}, 1\right]\), which is non-degenerate. As 
    \[\left(\hat{\mu}, 1\right] \cap \left[\underline{\mu}, 1\right) \cap \left[\underline{\mu}, 1\right] \neq \emptyset\text{,}\] we conclude that there exists a belief \(\mu \in \Delta\) and a utility function \(u \in \mathcal{U}\) such that pre-transformation it is uniquely optimal for the DM to choose action \(a^i\), yet \(a^i\) is strictly inferior to some \(a^j\) with \(j > i\) after the transformation; \textit{viz.,} the transformation does not reduce the DM's action.

    \bigskip

    We have argued that \(\hat{a}^i_\theta \leq a^i_\theta\) for all \(a^i \in A_{\theta,\theta'} \setminus \left\{\underline{a}^{\left(\theta,\theta'\right)}\right\}\) for all \(\theta, \theta' \in \Theta\), is necessary for a transformation to reduce the DM's action. It remains to prove the second half of the theorem; namely, to show that \(\hat{a}^i_{\theta'} \geq a^i_{\theta'}\) for all \(a^i \in A_{\theta,\theta'} \setminus \left\{\overline{a}^{\left(\theta,\theta'\right)}\right\}\) for all \(\theta, \theta' \in \Theta\), is necessary for a transformation to reduce the DM's action. This proof mimics the first part nearly exactly, so we provide fewer details.
    
    Suppose for the sake of contraposition that there exists a pair \(\theta, \theta' \in \Theta\) and an action \(a^{i} \in A_{\theta,\theta'} \setminus \left\{\overline{a}^{\left(\theta,\theta'\right)}\right\}\) such that \(a_{\theta'}^{i} > \hat{a}_{\theta'}^{i}\). 

    Let \[u(x) = \min\left\{x, \iota x + (1-\iota) a_{\theta'}^{i}\right\}\text{,}\]
    for some \(\iota \in \left(0,1\right)\). \(\underline{\mu}\) is as defined in Expression \ref{ex4} and remains strictly less than \(1\) for all \(\iota \in \left(0,1\right)\)--note that in contrast to the first part, these monetary values all lie below \(u\)'s inflection point. On the other hand, as \(\iota \downarrow 0\), 
    \[\overline{\mu}^{\dagger} \coloneqq \min_{j > i} \frac{u(a^{i}_{\theta}) - u(a^{j}_{\theta})}{u(a^{i}_{\theta}) - u(a^{j}_{\theta}) + u(a^j_{\theta'}) - u(a^{i}_{\theta'})} = \min_{j > i} \frac{a^{i}_{\theta} - a^{j}_{\theta}}{a^{i}_{\theta} - a^{j}_{\theta} + \iota\left(a^{j}_{\theta'} - a^{i}_{\theta'}\right)}\]
    converges to \(1\).
    
    Finally, 
    \[\hat{\mu}^{\dagger} \coloneqq \frac{u(\hat{a}^{i}_{\theta}) - u(\hat{a}^{i+1}_{\theta})}{u(\hat{a}^{i}_{\theta}) - u(\hat{a}^{i+1}_{\theta}) + u(\hat{a}^{i+1}_{\theta'}) - u(\hat{a}^{i}_{\theta'})} = \frac{\hat{a}^{i}_{\theta} - \hat{a}^{i+1}_{\theta}}{\hat{a}^{i}_{\theta} - \hat{a}^{i+1}_{\theta} + u(\hat{a}^{i+1}_{\theta'}) - \hat{a}^{i}_{\theta'}} < 1\text{,}\]
    for all \(\iota \in \left(0,1\right)\). We conclude that there exists a belief \(\mu \in \Delta\) and a utility function \(u \in \mathcal{U}\) such that pre-transformation it is uniquely optimal for the DM to choose action \(a^i\), yet \(a^i\) is strictly inferior to some \(a^j\) with \(j > i\) after the transformation; \textit{viz.,} the transformation does not reduce the DM's action.
\end{proof}

\subsection{Proposition \ref{regularprop} Proof}\label{regularpropproof}
\begin{proof}
    Let \(b \triangleright a\) and suppose \(a\) and \(b\) are made commonly steeper. Let \(a \succ b\). We want to show first that \(a \succ \hat{b}\). If \(b = \hat{b}\) this is immediate. Now let \(b\) dominate a mixture of \(a\) and \(\hat{b}\), and suppose for the sake of contradiction that \(\hat{b} \succeq a\). Then, for some \(\lambda \in [0,1)\)
    \[b \succ \lambda a + (1-\lambda) \hat{b} \succeq a\text{,}\]
    which is false. Accordingly, \(a \succ \hat{b}\). 
    
    Second, we want to show that \(\hat{a} \succ \hat{b}\). If \(\hat{a} = a\), the result is immediate. Now let \(\hat{a}\) dominate a mixture of \(a\) and \(\hat{b}\). Thus, for some \(\gamma \in \left[0,1\right]\),
    \[\hat{a} \succ \gamma a + (1-\gamma) \hat{b} \succeq \hat{b}\text{,}\]
    as desired.

    Now let \(a \sim b\). We want to show first that \(a \succeq \hat{b}\). To that end, we observe that if \(b = \hat{b}\), \(a \sim b \sim \hat{b}\), as desired; and if \(b\) dominates a mixture of \(a\) and \(\hat{b}\), for some \(\lambda \in [0,1]\)
    \[a \sim b \succ \lambda a + (1-\lambda) \hat{b}\text{,}\]
    so \(a \succ \hat{b}\). If \(\hat{a} = a\), \[\hat{a} \sim a \succeq \hat{b}\text{,}\] as desired. Now let \(\hat{a}\) dominate a mixture of \(a\) and \(\hat{b}\). Thus, for some \(\gamma \in \left[0,1\right]\),
    \[\hat{a} \succ \gamma a + (1-\gamma) \hat{b} \succeq \hat{b}\text{,}\]
    as desired.\end{proof}

\subsection{Verification of Regularity}\label{regsec}

Here we verify that the following four properties of the preferences over actions, \(\succeq\), hold for certain non-EU preferences. The properties to be verified are
\begin{enumerate}
    \item \textbf{Strongly Monotone:} If \(a \gg b\) then \(a \succ b\);
    \item \textbf{Convex:} If \(a \succeq b\) then for all \(\lambda \in [0,1]\), \(\lambda a + (1-\lambda) b \succeq b\);
    \item \textbf{Complete:} For all \(a, b \in \mathbb{R}^{\Theta}\), \(a \succeq b\) or \(b \succeq a\); and 
    \item \textbf{Transitive:} If \(a \succeq b\) and \(b \succeq c\) then \(a \succeq c\).
\end{enumerate}

\bigskip

\noindent \textbf{Variational preferences.} A DM solves \[\max_{a \in A} \min_{p \in \Delta}\left\{\mathbb{E}_p[u(a_\theta)] + c(p)\right\}\text{,}\]
where \(c \colon \Delta \to \mathbb{R}_{+}\) is a convex function and \(u\) is concave and strictly increasing.

We assume completeness, and transitivity is immediate. To prove strong monoticity suppose \(a \gg b\). Letting \(p_a\) and \(p_b\) be the respective minimizers of \[\mathbb{E}_p[u(a_\theta)] + c(p) \quad \text{and} \quad \mathbb{E}_p[u(b_\theta)] + c(p)\text{,}\]
we have
\[\mathbb{E}_{p_a}[u(a_\theta)] + c(p_a) > \mathbb{E}_{p_a}[u(b_\theta)] + c(p_a) \geq \mathbb{E}_{p_b}[u(b_\theta)] + c(p_b)\text{,}\]
as desired. To prove convexity, suppose 
\[\mathbb{E}_{p_a}[u(a_\theta)] + c(p_a) \geq \mathbb{E}_{p_b}[u(b_\theta)] + c(p_b)\text{.}\]
Given \(\lambda \in \left[0,1\right]\), let \(p^\lambda\) minimize
\[\mathbb{E}_p[u(\lambda a_\theta + (1-\lambda) b_\theta)] + c(p)\text{.}\]
Then, we have
\[\begin{split}
    \mathbb{E}_{p^{\lambda}}u(\lambda a_\theta + (1-\lambda) b_\theta) + c(p^{\lambda}) &\geq \lambda \mathbb{E}_{p^{\lambda}}u(a_\theta) + (1-\lambda) \mathbb{E}_{p^{\lambda}}u(b_\theta) + c(p^{\lambda})\\
    &\geq \lambda \left[\mathbb{E}_{p_a}u(a_\theta) + c(p_a)\right] + (1-\lambda) \left[\mathbb{E}_{p_b}u(b_\theta) + c(p_b)\right]\\
    &\geq \mathbb{E}_{p_b}[u(b_\theta)] + c(p_b)\text{.}
\end{split}\]

\bigskip

\noindent \textbf{Smooth ambiguity.} The DM evaluates an action according to 
\[\int_{\Delta} \varphi \left(\int_{\Theta} u(a(\theta))d\mu(\theta)\right)d\pi\text{,}\]
where \(u\) and \(\varphi\) are strictly increasing and \(u\) and \(\varphi\) are concave.

Completeness and transitivity are immediate. To verify strong monotonicity, let \(a \gg b\), in which case
\[\int_{\Delta} \varphi \left(\int_{\Theta} u(a(\theta))d\mu(\theta)\right)d\pi > \int_{\Delta} \varphi \left(\int_{\Theta} u(b(\theta))d\mu(\theta)\right)d\pi\text{,}\]
i.e., \(a \succ b\). To verify convexity, suppose \(a \succeq b\). We have
\[\begin{split}
    &\int_{\Delta} \varphi \left(\int_{\Theta} u(\lambda a(\theta) + \left(1-\lambda\right) b_\theta)d\mu(\theta)\right)d\pi\\ &\geq \lambda \int_{\Delta} \varphi \left(\int_{\Theta} u(a(\theta))d\mu(\theta)\right)d\pi + (1-\lambda) \int_{\Delta} \varphi \left(\int_{\Theta} u(b(\theta))d\mu(\theta)\right)d\pi\\
    &\geq \int_{\Delta} \varphi \left(\int_{\Theta} u(b(\theta))d\mu(\theta)\right)d\pi\text{.}
\end{split}\]

\subsection{Binary State Derivation}\label{binaryenviron}
We consider the binary environment of \(\S\)\ref{sectwostates}: \(\Theta = \left\{0,1\right\}\); and \(A = \left\{a^1, \dots, a^m\right\} \subset \mathbb{R}\), with \(a^1 \ \triangleleft \ \cdots \ \triangleleft \ a^m\), 
so that 
\[a^1_{0} > \dots > a^m_{0} \quad \text{and} \quad a^1_{1} < \dots < a^m_{1}\text{.}\]
We also assume that no action is weakly dominated for a risk-neutral agent. This guarantees \citep*{battigalli2016note,weinstein2016effect} that no action is weakly dominated for a risk-averse agent. For \(i \in \left\{1,\dots,m-1\right\}\) let \(\mu_i\) denote the indifference belief between actions \(a^i\) and \(a^{i+1}\) (pre-transformation). By assumption \(\mu_1 < \cdots < \mu_{m-1}\). We prove the following result.
\begin{proposition}
    A transformation reduces the DM's action if for any action \(a^i \in A\), \(a^i\) and \(a^{i+1}\) are made steeper.
\end{proposition}
\begin{proof}
    Suppose for every action \(a^i \in A\), \(a^i\) and \(a^{i+1}\) are made steeper. For all \(i \in \left\{1,\dots,m-1\right\}\) let \(\hat{\mu}_i\) denote the indifference belief between actions \(\hat{a}^i\) and \(\hat{a}^{i+1}\) (pre-transformation). By the increase in steepness \(\mu_i \leq \hat{\mu}_i\) for all \(i\).

    Observe that the transformation's monotonicity guarantees that the state-wise rankings of the actions are preserved. Hence, if the transformation does not shrink the set of rationalizable actions for a risk-neutral DM, we are done as \(\hat{\mu}_1, \dots, \hat{\mu}_{m-1}\) are the only relevant indifference beliefs for the DM's decision, and they all move in the ``right'' direction.

    The transformation, may, of course, shrink the set of rationalizable actions. Without loss of generality, we specialize to the first three actions (we know that there must be at least three in this case), and suppose for the sake of contradiction that the indifference belief between \(\hat{a}^3\) and \(\hat{a}^1\), \(\hat{\mu}_{\dagger}\) is weakly less than \(\mu_1\). This implies that \(\hat{\mu}_2 \leq \hat{\mu}^{\dagger}\), and so we have the chain
    \[\mu_1 < \mu_2 \leq \hat{\mu}_2 \leq \hat{\mu}^{\dagger}\text{,}\]
    a contradiction.\end{proof}

\end{document}